\documentclass[a4paper, UKenglish, cleveref, numberwithinsect, thm-restate]{lipics-v2021}
\pdfoutput=1 %uncomment to ensure pdflatex processing (mandatatory e.g. to submit to arXiv)
\hideLIPIcs  %uncomment to remove references to LIPIcs series (logo, DOI, ...), e.g. when preparing a pre-final version to be uploaded to arXiv or another public repository

\title{Continuous Algebras with Hypotheses}

\author{Lukas Mulder}{Radboud University, The Netherlands}{}{}{} % \and My second affiliation, Country \and \url{http://www.myhomepage.edu} }{johnqpublic@dummyuni.org}{https://orcid.org/0000-0002-1825-0097}{(Optional) author-specific funding acknowledgements}%TODO mandatory, please use full name; only 1 author per \author macro; first two parameters are mandatory, other parameters can be empty. Please provide at least the name of the affiliation and the country. The full address is optional. Use additional curly braces to indicate the correct name splitting when the last name consists of multiple name parts.

\author{Damien Pous}{CNRS, LIP, ENS de Lyon, France}{}{https://orcid.org/0000-0002-1220-4399}{} % {joanrpublic@dummycollege.org}{[orcid]}{[funding]}

\author{Jana Wagemaker}{Radboud University, The Netherlands}{}{https://orcid.org/0000-0002-8616-3905}{}

\authorrunning{L. Mulder, D. Pous, J. Wagemaker} %TODO mandatory. First: Use abbreviated first/middle names. Second (only in severe cases): Use first author plus 'et al.'

\Copyright{Lukas Mulder, Damien Pous and Jana Wagemaker} %TODO mandatory, please use full first names. LIPIcs license is "CC-BY";  http://creativecommons.org/licenses/by/3.0/

% TODO mandatory: Please choose ACM 2012 classifications from https://dl.acm.org/ccs/ccs_flat.cfm 


\category{} %optional, e.g. invited paper

\relatedversiondetails{This article appeared in Proc. CONCUR 2026 and the present version contains the appendices with all proofs}{https://doi.org/10.4230/LIPIcs.CONCUR.2026.28} 

\funding{This work was partially supported by the Dutch research council (NWO) under grant no.\ VI.Veni.242.134 (VerHyp) and by the LABEX MILYON (ANR-10-LABX-0070)}%optional, to capture a funding statement, which applies to all authors. Please enter author specific funding statements as fifth argument of the \author macro.

\acknowledgements{We thank Paul Brunet and Amina Doumane for early discussions on this work.}

\nolinenumbers %uncomment to disable line numbering

\EventEditors{Ana Sokolova and Patrick Totzke}
\EventNoEds{2}
\EventLongTitle{37th International Conference on Concurrency Theory (CONCUR 2026)}
\EventShortTitle{CONCUR 2026}
\EventAcronym{CONCUR}
\EventYear{2026}
\EventDate{September 1--4, 2026}
\EventLocation{Liverpool, UK}
\EventLogo{}
\SeriesVolume{391}
\ArticleNo{28}
\usepackage{macros}

\begin{document}

\maketitle

\begin{abstract}
In the literature on Kleene algebra (KA), a number of variants have been proposed such as Kleene algebra with tests, commutative KA, bi-KA, and concurrent KA. 
The equational theories of some of these structures have then been studied in the presence of additional assumptions, called \emph{hypotheses}.
We propose a unifying framework encompassing all the previous structures, as well as regular tree languages. This is done by considering algebras ordered by complete lattices, where least fixpoints can be computed. We provide a canonical model consisting of closed languages, which we prove sound and complete with respect to all \emph{continuous} models. Then we study quasi-equational axiomatisations. It is illusory to hope for a generic axiomatisation which would be sound and complete for all instances. Instead, we provide a generic axiomatisation which we prove sound and we setup tools that make it possible to get complete ones in a modular way, building on previous works from the literature. We showcase these tools by proving new completeness results for commutative KA, bi-KA, and regular tree languages, in each case extended with various hypotheses.
\end{abstract}

\section{Introduction}
\label{sec:intro}

Kleene algebras (KA)~\cite{kleene1956representation,conway1971regular} are 
% algebraic 
structures with sequential composition, iteration and choice. They were first proposed as the algebra of regular expressions, and their axioms are sound and complete w.r.t. relational models and language models~\cite{boffa1990remarque, Kozen91, Krob91a}. Their equational theory is decidable via automata algorithms, which makes it possible to design automation tactics in proof assistants~\cite{bp:itp10:kacoq,krauss2012pearl,pous2013coq}. In such applications, we often want to reason under assumptions, or hypotheses. 
% For instance, to deduce that a given equation on relations holds provided some of them are transitive, or can be permuted. 
Unfortunately, the Horn theory of Kleene algebras is undecidable~\cite{Kozen02,Kuznetsov23} and we may only automate reasoning under specific forms of hypotheses~\cite{cohen94:ka:hypotheses,KozenS96,hardin2002elimination,kozen2014equations}.

In each case, proving completeness and extending decidability is a delicate issue. This has motivated the development of a general theory of \emph{Kleene algebra with hypotheses}~\cite{kozen2014equations,dkpp:fossacs19:kah,pous2024tools}. There, models of \emph{closed languages} are shown to characterise the Horn theory of \emph{star-continuous} Kleene algebras, and a notion of \emph{reduction} makes it possible to approach completeness and decidability proofs in a principled and modular way. 

Beside hypotheses, Kleene algebras have also been extended to deal with common programming or logical constructs: conditionals in Kleene algebra with tests (KAT)~\cite{kozen1997tests}, mutable tests~\cite{GKM14a:KATB}, concurrency in bi-KA~\cite{struth2014bikleene}, concurrent KA~\cite{hoare2011concurrent} and synchronous KA (SKA)~\cite{Prisacariu10,WagemakerBKR019},  intersection~\cite{dp:concur18:kl}, or converse~\cite{EB95,bloom1995notes}. %, residuals~\cite{Pratt90}.
Here again, completeness proofs are notoriously hard, and while some of the above extensions do fit---with some effort---into the framework of KA with hypotheses (e.g., KAT, SKA and KA with converse), some of them do not, typically when operations such as parallel composition or intersection are added, or when the required closures do not preserve regularity.
An extended framework was developed~\cite{KappeB0WZ20, kappethesis, wagemakerthesis} to deal with hypotheses in concurrent versions of KA, but a unifying theory is still lacking.

We propose such a theory in the present paper.
Before introducing this abstract theory, observe that the operations from KA can be organised into three categories:
\begin{enumerate}
\item dot $(\cdot)$ and one $(1)$, which form a monoid, % and yield words in the semanics;
\item plus $(+)$ and zero $(0)$, which form a semilattice, % and make it possible to obtain (finite) sets of words;
\item star, which is the least fixpoint of a function expressible from ${\cdot}, {+}$ and $1$ $(e^*=\mu x.1+e\cdot x)$.
\end{enumerate}
When moving from one variant of KA to another, we keep the semilattice operations, but we change the first layer and the set of allowed fixpoints. For instance, to get bi-KA, we use bimonoids and two fixpoint operators; to get commutative KA~\cite{conway1971regular, brunet2019commutative}, we simply require the monoid to be commutative. Moreover, the standard language models from the literature consist of sets of elements of the free algebraic structure from the first layer: words for KA, series-parallel pomsets for bi-KA, finite multisets for commutative KA.

Accordingly, our theory is parameterised by an arbitrary signature $\Sigma$ and a set of equations $E$ for the first layer. We define our syntax by adding plus and zero to the operations from $\Sigma$, as well as arbitrary fixpoints. We interpret the resulting expressions in \emph{continuous models}, which are models of $E$ over complete lattices and where all operations from $\Sigma$ are continuous. In particular, we get such a model by considering languages of \emph{atoms}: terms built only from $\Sigma$ and quotiented by $E$. 
The main requirement here is that the equations in $E$ should be linear and regular: each variable must appear exactly once on each side of an equation, so that the algebra of languages we obtain remains a model of $E$~\cite[Theorem~3]{shafaat1974varieties}. 
We show that every interpretation factors through this language model, witnessing its canonicity~(\Cref{thm:interpretation factors}).

At this point we have captured all the variants discussed above in a uniform way (and more, such as  regular tree languages~\cite{esik2010tree} by taking the empty set for $E$).
The next step consists of dealing with hypotheses. Following and generalising~\cite{dkpp:fossacs19:kah,KappeB0WZ20,pous2024tools}, we associate to every set $H$ of hypotheses a closure operator $H^*$ on languages. This operator makes it possible to refine the former language model to take hypotheses into account, and we prove the following characterisation (\Cref{thm:sound:complete}):
\begin{align*}
  \CL\models H\rightarrow e\leq f \iff \cinterp e _E \subseteq \cinterp f _E\enspace.
\end{align*}
In words, an inequation $e\leq f$ between expressions follows from a set $H$ of hypotheses in all continuous models of $E$ if and only if the $H$-closed language of $e$ is contained in that of $f$.
This characterisation is powerful as it reduces the validity of a Horn sentence in all continuous models to a language-theoretic problem. For instance, the right-hand side of the above equivalence often turns out to be decidable via automata algorithms.

It remains to study potential axiomatisations. It is easy to come up with quasi-equational axiomatisations $Q$ which are \emph{sound} for all continuous models:
\begin{align*}
  Q\proves H\rightarrow e\leq f \implies \CL\models H\rightarrow e\leq f\enspace.
\end{align*}
In particular, the ``naive'' axiomatisation consisting of the equations in $E$, the fact that symbols from $\Sigma$ distribute over finite joins, and the fact that fixpoint expressions are indeed least fixpoints, is always sound (\Cref{prop:rules:sound,lem:naive:sound}).

Obtaining the converse implication for all Horn sentences (i.e., \emph{Horn completeness}) is difficult, if not impossible~\cite{Kozen02,kozen2014equations}. Thus we focus on specific sets $H$ of hypotheses: we say that an axiomatisation $Q$ is \emph{complete for $H$ (on $E$)} if for all expressions $e,f$, we have
\begin{align*}
  Q\proves H\rightarrow e\leq f \implied \CL\models H\rightarrow e\leq f\enspace.
\end{align*}
Actually, completeness is difficult to achieve even when $H$ is empty.
Ésik has proved that the aforementioned naive axiomatisation is complete when $E$ is also empty~\cite{esik2010tree}.
Unfortunately, this axiomatisation is incomplete on monoids: there we need the axioms of Conway algebras~\cite{conway1971regular,Krob91a}, (left-handed) Kleene algebras~\cite{Kozen91,KozenS12,ddp:lpar18:lefthanded,Kappe23}, or Boffa algebras~\cite{Boffa95}.
Still, we can import all these completeness results for the empty set of hypotheses in our framework, as well as the ones for 
bi-KA on series-parallel pomsets~\cite{struth2014bikleene}, commutative KA on finite multisets~\cite{pilling1970algebra,conway1971regular}, KA with converse on involutive words~\cite{EB95}, and identity-free distributive Kleene lattices on series-parallel graphs~\cite{dp:concur18:kl}.

There are also completeness results for specific non-empty sets of hypotheses, which fit our framework:
on monoids, KA is complete for $\set{e\leq 0}$~\cite{cohen94:ka:hypotheses,kozen1997tests}, for 
$\set{a\leq w}$~\cite{kozen2014equations,dkpp:fossacs19:kah} (with $a$ one or a letter, and $w$ a word), for $\set{S=1}$~\cite{dkpp:fossacs19:kah} (with $S$ a sum of letters), and for $\set{e\leq \top\mid \forall e}$ and $\set{e\leq \top,~e\leq e\top e\mid \forall e}$~\cite{pous2022top,PousW24};
on involutive monoids, KA with converse is complete for $\set{e\leq ee^\circ e\mid \forall e}$~\cite{EB95};
on pomsets, bi-KA is complete for $\set{a\leq w\mid |w|\geq 1}$ (with $a$ a letter, $w$ a word) and for the \emph{exchange law} $\set{(e\parallel f)\cdot (g\parallel h)\leq (e\cdot g)\parallel (f\cdot h)\mid \forall e,f,g,h}$~\cite{kappethesis}, which yields completeness of concurrent KA~\cite{kappe2018concurrent}.
% + easy ones ($a\leq 1$, $aa\leq a$, $a^\circ\leq a$) ? last one probably not in the literature
%
To deal with the diversity of these completeness results and be able to combine some of them together, we follow the literature and develop an appropriate notion of \emph{reduction}~\cite{KappeB0WZ20,prw:ramics21:mkah,pous2024tools,brunet26:representations}.

\subparagraph*{Contributions} 

We propose a framework for studying least fixpoints in algebras on complete lattices. We propose a model of languages which is canonical for continuous models of sets of linear and regular equations (\Cref{sec:framework}). We show how add hypotheses in such structures, by defining a notion of closed languages that capture their Horn theory (\Cref{sec:closure}). We finally study axiomatisations (\Cref{sec:axiomatisation}): we provide a naive axiomatisation which is always sound, as well as tools to establish completeness, encompassing all results from the literature (\Cref{sec:reductions}). We illustrate our framework by proving new completeness results for commutative KA, bi-KA,
and regular tree languages in the presence of specific hypotheses (\Cref{sec:examples}). 

\subparagraph*{On the distinction between equations and hypotheses}
In the proposed framework, equations $(E)$ and hypotheses $(H)$ are quite close from the end-user perspective. The main differences are the following:
\begin{itemize}
\item equations are rather constrained (universally quantified, only between terms, linear regular, not inequations), but they can be integrated into the language semantics in a direct manner;
\item in contrast, hypotheses are arbitrary (possibly about specific variables, between expressions, not necessarily linear/regular, can be inequations), but they have to be dealt with using the machinery of closed languages.
\end{itemize}

% 
%% DAMIEN: not yet
% Besides the generality of the framework, an important contribution is that it makes it possible to study hypotheses at the right abstraction level, and to obtain reductions from one abstraction level to another.
% TODO: makes this clear with examples (various formulations of KAT atoms and guarded strings ; Kleene algebra with converse ; need E,E' reductions?)
%

\section{Preliminaries}
\label{sec:preliminaries}

Given a set $X$, we write $\pset X$ for the set of its subsets (its \emph{powerset}), and $X^n$ for its $n$-fold Cartesian product.
Given a partial map $f$, an element $x$ and a value $y$, we write $\override f x y$ for the partial map whose value is $y$ at $x$, and the same as $f$ on the other elements. 
We also may write $\vec x = x_1,\dots, x_n$ for a sequence of elements of a set $X$.

An \emph{algebra} for a signature $\Sigma$ is a set $A$ together with a function $s_A\colon A^n\to A$ for each symbol of arity $n$ in $\Sigma$.
Symbols of arity zero are \emph{constants}; they are just elements in the algebras.
When it is clear from the context, we omit the subscript in $s_A$.
The \emph{powerset lifting} $\pset A$ of such an algebra $A$ is the algebra obtained by lifting its operations to sets in a pointwise manner: for a symbol $s$ of arity $n$ and for all $L_i \in \pset A$ we set
\begin{align*}
  s_{\pset A}(L_1, \dots, L_n) = \{ s_A(a_1, \dots, a_n) \mid \forall i, a_i \in L_i \}\enspace.
\end{align*}
Given two algebras $A$ and $B$, a \emph{homomorphism from $A$ to $B$} is a function $f \colon A \to B$ such that $f(s_A(a_1,\dots, a_n)) = s_B(f(a_1), \dots, f(a_n))$ for all symbols $s$ of arity $n$ and all elements $a_1,\dots,a_n$ from $A$.

A \emph{partial order} is a set $S$ with a binary relation $\leq$ which is reflexive, transitive, and antisymmetric.
% We write $\vec x=x_1,\dots,x_n$ for an ordered sequence, we write $\vec x\leq\vec y$ when $n=m$ and $x_i\leq y_i$ for all indices $i$.
The functions from a set into a partial order $S$, are partially ordered pointwise ($f\leq g\eqdef \forall x,~f(x)\leq g(x)$). 
A function $f\colon S \to S'$ between two partial orders is \emph{monotone} if $f(x)\leq f(y)$ in $S'$ whenever $x\leq y$ in $S$.
% The subset of such monotone functions is written $\hom{S}{S'}$.
A \emph{partially ordered algebra} is an algebra whose carrier is a partial order and whose  operations are monotone.

A \emph{semilattice} is a partial order $S$ with a least element $(\perp)$ and all binary joins $(\vee)$.
If every subset $S'$ of $S$ has a supremum $\bigvee S'$, then the semilattice is called a \emph{complete lattice}, and we denote it by $(S, \bigvee)$. Every monotone function $f$ on a complete lattice admits a least fixpoint, written $\mu f$, which is also its least pre-fixpoint~\cite{tarski1955lattice}.
The powerset $\pset X$ of a set $X$, ordered by inclusion, is the free complete lattice on $X$ in the sense that every function $f\colon X\to S$ from $X$ to a complete lattice $S$ extends to a (unique) complete lattice homomorphism $\bigvee f\colon\pset X\to S$ by setting $\bigvee f(Y) \eqdef \bigvee_{y\in Y} f(y)$.
A function $f$ between complete lattices is \emph{continuous} if it preserves all suprema: $f\bigvee Y=\bigvee_{x\in Y}f(x)$ for all $Y$.

\section{Framework}
\label{sec:framework}
\pratendSetLocal{category=framework}

\subsection{Terms and expressions}
\label{ssec:terms:expressions}
We fix throughout the paper a signature $\Sigma$ and a set $X$ of \emph{variables}.
\begin{definition}\label{def:term}
  We write $\term{X}$ for the set of \emph{terms} with variables in $X$, which are defined by the following grammar.
  \begin{align*}
    t ::= x \mid s(t_1, \dots, t_n) \tag{$x \in X,~s^{(n)} \in\Sigma$}
  \end{align*}
\end{definition}
Note that terms are trees whose leaves are either constants from $\Sigma$ or variables from $X$;
they form the free algebra over $X$, in the sense that every \emph{valuation} $\sigma\colon X\to A$ into an algebra $A$ for $\Sigma$ extends to a unique homomorphism $\ext\sigma\colon\term{X}\to A$.

We extend terms to \emph{expressions} by adding a constant $0$, a binary symbol $+$, and a least fixpoint operator $\mu$. For the latter, we assume a second set $R$ of \emph{recursion variables}, countably infinite and disjoint from $X$.
\begin{definition}\label{def:expression}
  \emph{Expressions} are defined by the following grammar
  \begin{align*}
    e,f ::= x \mid s(e_1,\dots,e_n) \mid 0 \mid e+f \mid \mu y.e \tag{$x \in X\cup R$, $y\in R$, $s^{(n)}\in\Sigma$}
  \end{align*}
  The recursion variable $y$ in the expression $\mu y.e$ is \emph{bound} in $e$.
  An expression is \emph{closed} if all its recursion variables are bound.
  An \emph{$x$-expression} is an expression containing at most one unbound recursion variable, $x$.
  We denote the set of closed expressions by $\expr{X}$.
\end{definition}
Every term can be seen as a closed expression: we have a canonical embedding $\term{X} \hookrightarrow \expr{X}$.
Also note that closed expressions may contain variables in $X$, which are never bound, and that $x$-expressions over $X$ can also be seen as closed expressions over $X \cup \set{x}$.
This change of perspective is necessary for some proofs found in the appendix.
However, throughout the main text we keep $X$ fixed and omit it from notation for the sake of clarity.

A \emph{(closed) substitution} is a partial map from variables in $X\cup R$ to closed expressions; given an expression $e$ and a substitution $\rho$, we write $e\rho$ for the \emph{application} of $\rho$ to $e$:
\phantomsection\label{page:subst}
\begin{align*}
  \begin{aligned}
    0\rho &\eqdef 0\\
    (e+f)\rho &\eqdef e\rho + f \rho\\
    (s(e_1, \dots, e_n))\rho &\eqdef s(e_1\rho, \dots, e_n\rho)\\
  \end{aligned}&&
  \begin{aligned}
    x\rho &\eqdef
              \begin{cases}
                \rho(x)&\text{if $\rho(x)$ is defined,}\\
                x&\text{otherwise}\\
              \end{cases}
    \\
    (\mu x.e)\rho &\eqdef \mu x.e(\rho\backslash x)\enspace.
  \end{aligned}
\end{align*}
(In the last case, $\rho\backslash x$ is the substitution obtained from $\rho$ by removing its value at $x$, if any.)
The above definition captures at once the two main kinds of substitutions we need:
\begin{itemize}
\item we let $\theta$ range over substitutions of variables from $X$ only; 
\item we write $\sub{e}{x}{f}$ for substitutions of an $x$-expression $e$ where the recursion variable $x\in R$ is substituted for $f$, resulting in a closed expression.
\end{itemize}
\medskip
In order to interpret expressions, we use algebras ordered by complete lattices:
\begin{definition}
  A \emph{complete algebra} is an algebra partially ordered on a complete lattice.
\end{definition}

\begin{definition}\label{def:interpretation}
  Let $A$ be a complete algebra.  To every expression $e$, we
  associate a monotone function
  $\evalb{\_}{e}\colon ((X\cup R) \to A)\to A$, by structural
  induction on $e$:
  \begin{align*}
    % [\_]_{\_} \colon \expr{X} &\to \hom{(X \to \L)}{\L} \\
    \begin{aligned}
      \evalb\rho{0} &\eqdef {\perp}\\
      \evalb\rho{e+f}&\eqdef\evalb\rho{e} \vee \evalb\rho{f}\\
      \evalb\rho{s(e_1, \dots, e_n)} &\eqdef s_A(\evalb\rho{e_1}, \dots, \evalb\rho{e_n})\\
    \end{aligned}&&
    \begin{aligned}
      \evalb\rho{x} &\eqdef \rho(x)\\
      \evalb\rho{\mu x.e} &\eqdef \mu a.\evalb{\override\rho x a} e\enspace.
    \end{aligned}
  \end{align*}
  Given a valuation $\sigma\colon X\to A$ and a closed expression $e$, we write $\eval\sigma e$ for $\evalb{\sigma,R\mapsto\perp} e$, which we call the \emph{interpretation of $e$ under $\sigma$}.
\end{definition}
In $\evalb\rho{\mu x.e}$, $a \mapsto \evalb{\override\rho x a} e$ is a monotone function by induction; hence we can take its least fixed point in $A$.
Moreover, since taking least fixed points is a monotone operation, $\rho\mapsto \mu a.\evalb{\override\rho x a} e$ is also monotone, as required.
We overload $\ext\sigma$: when $\sigma$ is a valuation into a complete algebra, $\ext\sigma$ can be applied to terms and closed expressions; this is safe thanks to the following lemma:

\begin{lemmaE}
  \label{lem:interpretation:term:expr}
  \prooflink
  For all valuations $\sigma \colon X \to A$ into a complete algebra $A$, the following diagram commutes.
  \[\begin{tikzcd}[ampersand replacement=\&]
      {\term{X}} \& {\expr{X}} \\
      \& A
      \arrow[hook, from=1-1, to=1-2]
      \arrow["\ext\sigma", from=1-2, to=2-2]
      \arrow["\ext\sigma"', from=1-1, to=2-2]
    \end{tikzcd}\]
\end{lemmaE}
\begin{proofE}
  Both functions from $\term{X}$ to $A$ are homomorphisms and they agree on variables.
  The statement follows by unicity of homomorphic extensions.
\end{proofE}
In the sequel we just call \emph{expressions} the closed expressions. Conversely, we systematically call \emph{open expressions} the expressions which are not assumed to be closed.

\subsection{Models}
\label{ssec:models}

We distinguish terms and expressions because we will interpret expressions as languages of terms quotiented by a set $E$ of equations.
% For instance, to get Kleene algebras and standard word languages, we will use the signature $\set{\cdot^{(2)},1^{(0)}}$ together with the monoid equations.
%
A \emph{(term) equation} is a pair of terms written $u=v$. Variables in equations are universally quantified: $u=v$ is \emph{satisfied} in an algebra $A$ if for all valuations $\sigma\colon X\to A$ in that algebra, we have $\eval\sigma{u}=\eval\sigma{v}$.
A \emph{model} of a set of equations is an algebra satisfying these equations.
% \begin{gray}
%   The class of models of a given set of equations is also known as a \emph{variety}.
% \end{gray}
\begin{example}
  \label{ex:models}
  \emph{Monoids} are the class of models given by the signature
  $\set{\cdot^{(2)},1^{(0)}}$ and the equations $\set{x \cdot 1 = x, 1 \cdot x = x, x \cdot (y \cdot z) = (x \cdot y) \cdot z}$.
  If we add commutativity $(x \cdot y = y \cdot x)$ to the previous equations, we obtain \emph{commutative monoids}. \emph{Bimonoids} \cite{bloom1996shuffle,struth2014bikleene} are given by the signature $\set{\cdot^{(2)},\parallel^{(2)},1^{(0)}}$ together with monoid equations for $(\cdot, 1)$ and the commutative monoid equations for $(\parallel, 1)$.
  When $E$ is empty, models are arbitrary algebras. 
  % Finally, we remark that for an arbitrary signature $\Sigma$ without any equations, the resulting class of models are precisely the \emph{algebras}.
\end{example}

An equation $u=v$ is \emph{linear} if every variable appears at most once on either side;
it is \emph{regular} if every variable appearing in $u$ also appears in $v$ and vice versa.
\begin{theorem}[{\cite[Theorem 3]{shafaat1974varieties},\cite{gautam1957validity, bleicher1973permanence}}]\label{thm:linreg:pset}
  The class of models of a set $E$ of equations is closed
  under powerset lifting if and only if all equations in $E$ are
  linear and regular.
\end{theorem}

\begin{example}
  Monoids, commutative monoids and bimonoids are defined by linear and regular equations, hence these classes of models are closed under powerset lifting.
  % This is also trivially true for trees, since they have no equations. We provide counter-examples illustrating the need for both linearity and regularity in Appendix~\ref{app:framework}.
\end{example}

\begin{textAtEnd}
  To illustrate \Cref{thm:linreg:pset}, we provide two counter-examples showing that linearity and regularity of $E$ are required to obtain that the powerset lifting of a model of $E$ remains a model of $E$.
  \begin{counterexample}
    Consider a binary symbol $\vee$ and the idempotency equation $x \vee x = x$. This equation is not linear. Consider any idempotent algebra with two distinct elements $a$ and $b$, and such that $a\vee b\not\in\set{a,b}$ (e.g., sets of natural numbers with union, and any two distinct singletons). 
    Its powerset lifting is not idempotent: we have $\set{a,b}\vee\set{a,b}=\set{a\vee a,a\vee b,b\vee a,b\vee b}=\set{a,a\vee b,b\vee a,b}\neq\set{a,b}$. 
  \end{counterexample}

  \begin{counterexample}
    Consider the empty signature together with the non-regular equation $x = y$, stating every algebra has at most one distinct element.
    An algebra for this equation is the one element set $1 = \set{*}$.
    Its powerset lifting $\pow(1) = \set{\emptyset, 1}$ has two distinct elements, hence does not satisfy $x = y$.
  \end{counterexample}
\end{textAtEnd}

In addition to the signature $\Sigma$ and the variables $X$, we fix in the sequel a set $E$ of linear and regular term equations. When used alone, \emph{model} implicitly means model of $E$.
% Sometimes we use \emph{models} to refer implicitly to models of $E$.

The equations in $E$ generate an equational theory $\equiv_E$ on terms: the least congruence relation containing the closure of $E$ under arbitrary term substitutions; it coincides with the set of equations satisfied in all models of $E$.
We write $\atom X_E$ for  the set of terms quotiented by $\equiv_E$, called \emph{atoms}.
Atoms form the free model of $E$ over $X$: $\atom X_E$ is a model of $E$ and every valuation $\sigma\colon X\to M$ into a model $M$ of $E$ extends to a unique homomorphism $\ext\sigma\colon\atom X_E\to M$.

\begin{example}
    In the case of monoids, atoms are words with letters in $X$.
    For commutative monoids, atoms are finite multisets over $X$.
    For bimonoids, atoms are \emph{series-parallel pomsets (partially ordered multiset)}~\cite[Lemma~8]{struth2014bikleene}.
    Finally, when $E$ is empty, atoms are just terms.
\end{example}

Note that the $\ext\sigma$ notation is overloaded again: when $\sigma$ is a valuation into a model, $\ext\sigma$ can be applied to both terms and atoms, which is safe thanks to the following lemma.
\begin{lemmaE}
    \label{lem:interpretation:term:atom}
    \prooflink
    For all valuations $\sigma \colon X \to M$ into a model $M$ of $E$, the following diagram commutes:
    \[\begin{tikzcd}[ampersand replacement=\&]
        {\term{X}} \& \\
        {\atom X_E} \& M
        \arrow[two heads, from=1-1, to=2-1]
        \arrow["{\ext\sigma}", from=1-1, to=2-2]
        \arrow["{\ext\sigma}"', from=2-1, to=2-2]
    \end{tikzcd}\]
\end{lemmaE}
\begin{proofE}
  Both functions $\term X \to M$ in the diagram above are homomorphisms which agree on $X$, hence they must coincide by unicity of homomorphic extensions.
\end{proofE}
Given an atom $a$, we write $\underline a$ for some term in its equivalence class. We use this only in contexts where the specific choice of term is irrelevant, typically thanks to the above lemma.

Since $E$ is linear and regular, $\pset{\atom X_E}$ is a model of $E$ by \Cref{thm:linreg:pset}. Its elements are sets of atoms, which we call \emph{languages}. The notion of language thus depends on $E$: e.g. for monoids, where atoms are words, languages are sets of words.
We abbreviate $\pset{\atom X _E}$ as $\lang X _E$.

\begin{definition}
  A \emph{continuous model of $E$} is a complete algebra which is a model of $E$ and whose operations are continuous in each argument: for all symbols $s$ of arity $n+1$, for all elements $a_1,\dots, a_n$, and for all subsets $L$ of elements, we have
%  \damien{remove since continuous is defined now?}
  \begin{align*}
    s(a_1,\dots,\bigvee L,\dots,a_n) = \bigvee_{a\in L} s(a_1,\dots,a,\dots,a_n)\enspace.
  \end{align*}
\end{definition}
\vspace{-.8\baselineskip}%DAMIEN: why do we get a spurious line??
\begin{propositionE}
  \prooflink
  For all models $M$ of $E$, % the powerset lifting
  $\pset M$ is a continuous model of $E$.
\end{propositionE}
%% APPENDIX {powersetiscontinuousalgebra}
\begin{proofE}
  $\pow(M)$ is a model of $E$ by~\Cref{thm:linreg:pset} since we have assumed that $E$ is linear and regular. Moreover, by definition of the algebra structure on $\pow(M)$, all operations preserve arbitrary unions in each argument.
\end{proofE}
In particular, $\lang X_E$ is a continuous model of $E$, which we call the \emph{standard language model of $E$}. Expressions can be interpreted in this model by mapping variables to singleton languages.

\begin{definition}
    \label{def:standard language interpretation}
    The \emph{standard language interpretation} of an expression $e$ is the language defined by
    $\interp{e}_E \eqdef \eval\eta{e}$ for $\eta\colon x \mapsto \set x$, where $x$ in $\{x\}$ is the equivalence class of $x$ in $\atom X_E$.
\end{definition}
In the case of monoids, when restricting fixpoint expressions to be of the shape $\mu x.1+e\cdot x$, the function $\interp{\_}_E$ coincides with the standard language interpretation of regular expressions.

Given a valuation $\sigma \colon X \to M$ into a complete model $M$, we have defined its extension $\ext\sigma$ to both expressions and atoms; we also define its extension to languages, as follows:
\begin{align*}
  \lext\sigma\colon \lang X_E&\to M\\
  L&\mapsto \bigvee_{a\in L}\ext\sigma(a)
\end{align*}
We establish a few properties of such extensions to languages in Appendix~\ref{app:framework}; in particular:
\begin{textAtEnd}
  Before proving \Cref{thm:interpretation factors}, we establish a few properties about extensions $\lext\sigma$ of valuations $\sigma$ to languages.
\end{textAtEnd}
\begin{propositionE}
  \label{lem:sigma:is:an:algebra:morphism}
  \prooflink
  For all valuations $\sigma \colon X \to M$ into a continuous model $M$, the function $\lext\sigma$ is a continuous homomorphism from $\lang X_E$ to $M$.
\end{propositionE}
\begin{proofE}
  Continuity is straightforward, and for all symbols $s$ of arity $n$ we have
  \begin{align*}
    \lext\sigma(s_{\lang X}(L_1, \dots, L_n))
    &= \bigvee_{a \in s_{\lang X}(L_1, \dots, L_n)} \ext{\sigma}(a)\\
    &= \bigvee_{a_i \in L_i} \ext\sigma(s_{\atom X}(a_1, \dots, a_n))\\
    &= \bigvee_{a_i \in L_i} s_M(\ext\sigma(a_1), \dots, \ext\sigma(a_n))\\
    \tag{continuity of $M$}
    &= s_M(\bigvee_{a_1 \in L_1} \ext\sigma(a_1), \dots, \bigvee_{a_n \in L_n} \ext\sigma(a_n))\\
    \tag*\qedhere
    &= s_M(\lext\sigma(L_1), \dots, \lext\sigma(L_n))\enspace.
  \end{align*}
\end{proofE}

\begin{textAtEnd}
  \begin{lemma}
      \label{lem:language substitution commutes with join}
      For all sets of variables $Y$, for all valuations $\sigma' \colon X \to \lang Y$ into languages and for all valuations $\sigma \colon Y \to M$ into a continuous model $M$, we have:
      \begin{align*}
        \lext\sigma \circ \lext{\sigma'} = \lext{\sigma''}
        \quad\text{where}\quad \sigma''=\lext\sigma\circ\sigma'\enspace.
      \end{align*}
  \end{lemma}
  \begin{proof}
    Let $L$ be a language in $\lang X$. Unfolding definitions, we have
    \begin{align*}
      \lext{\sigma} (\lext{\sigma'}(L))
      &= \lext{\sigma} (\bigcup_{w \in L}  \ext{\sigma'}(w)) 
        = \bigvee_{w \in L} \bigvee_{a \in \ext{\sigma'}(w)}\ext\sigma(a)
        = \bigvee_{w \in L} \lext\sigma(\ext{\sigma'}(w))
      \\
      \lext{\sigma''}(L) &= \bigvee_{w \in L}\ext{\sigma''}(w)
    \end{align*}
    Hence, it suffices to verify that $\lext\sigma\circ\ext{\sigma'}=\ext{\sigma''}$.
    These functions are homomorphisms by definition and \Cref{lem:sigma:is:an:algebra:morphism}, and they agree on $X$ by definition, thus they are equal. 
  \end{proof}

  \longnotations

  Below we write $\eta_X$ for the function $x\mapsto \set x$ creating singleton sets from elements of $X$. This function has type $X\to \pset X$, we also freely use it as a function of type $X\to \pow (\expr X)$, or $X\to \lang X$, implicitly using the injections from variables into expressions and atoms, respectively. With this convention, we have $\interp e=\eval{\eta_X}e$ by definition.

  Given two languages $L\in \lang{X\uplus\set x}$ and $L'\in\lang Y$, we write $\sub{L}{x}{L'}$ for  $\lext{(\eta_X, x \mapsto L')}(L)$, which is a language in $\lang{X\cup Y}$. We call this operation language substitution. As an instance of the previous lemma, we obtain:
  \begin{corollary}
    \label{cor:language substitution commutes with join}
    For all sets of variables $Y$, for all languages $L\in \lang {X \uplus \set{x}}$, for all valuations $\sigma \colon X\cup Y \to M$ into a continuous model $M$ and for all languages $L'\in \lang Y$, we have
    \begin{align*}
      \lext\sigma(\sub{L}{x}{L'}) = \lext{\sigma'}(L)
      \quad\text{where}\quad\sigma'=\override\sigma x {\lext\sigma(L')}\enspace.
    \end{align*}    
  \end{corollary}

  We first prove the right-to-left inequation from \Cref{thm:interpretation factors}.
  Observe that this part of the statement is equivalent to $``\forall a\in\interp e,~\ext\sigma a\leq \ext\sigma e$'', which we get as an instance of the following lemma, by choosing $Y=X$ and $\rho=\eta_X$.
  \begin{lemma}
    \label{thm:key1}
    Let $\sigma\colon Y\to M$ be a valuation from a set $Y$ of variables into a continuous model $M$.
    For all sets $X$ of variables, for all expressions $e\in \expr X$, and for all valuations $\rho\colon X\to\lang Y$, we have
    \begin{align*}
      \forall a\in\eval\rho e,~ \ext\sigma(a)\leq \eval\theta e
      \quad\text{where}\quad\theta=\lext\sigma\circ\rho\enspace.
    \end{align*}    
  \end{lemma}
  \begin{proof}
    Fix $Y$, $M$, and $\sigma\colon Y\to M$.
    We proceed by induction and case analysis on $e$.
    \begin{itemize}
    \item if $e$ is a variable $x$ then $\eval\theta e=\lext\sigma(\rho(x))=\bigvee_{a\in \rho(x)}\ext\sigma(a)$, whence $\ext\sigma(a)\leq \eval\theta e$ for all $a\in\eval\rho e=\rho(x)$;
    \item if $e=0$ then the statement is vacuously true since $\eval\rho e=\emptyset$;
    \item if $e=e_1+e_2$ then an atom in $\eval\rho e$ belongs either to $\eval\rho {e_1}$ or $\eval\rho {e_2}$, and we use the induction hypothesis accordingly;
    \item if $e=s(e_1,\dots,e_n)$, then let $a\in\eval\rho e$; by definition, we have $a=s(a_1,\dots,a_n)$ with $a_i\in\eval\rho {e_i}$ for all $i$; then we just compute:
      \begin{align*}
        \tag{$\ext\sigma$ is a homomorphism}
        \ext\sigma(a) &= s(\ext\sigma(a_1),\dots,\ext\sigma(a_n))\\
        \tag{by induction hypothesis}
                      &\leq s(\eval\theta {e_1},\dots,\eval\theta {e_n})\\
        \tag{$\ext\theta$ is a homomorphism}
                      &= \eval\theta {s({e_1},\dots,{e_n})} = \eval\theta e
      \end{align*}      
    \item finally, if $e=\mu x.e'$ for an $x$-expression $e'$, then let $L=\{a\mid \ext\sigma(a)\leq \eval\theta e\}$ so that we have to show $\eval\rho{\mu x.e'}\subseteq L$.
      Since $\eval\rho{\mu x.e'}$ is a least fixpoint, it suffices to show $\eval{\rho'}{e'}\subseteq L$, where $\rho'=\rho,x\mapsto L$.
      By induction hypothesis on $e'$ (seen as an expression on $X\cup\set x$), we get $\eval{\rho'}{e'}\subseteq L'$ where $L'=\{a\mid \ext\sigma(a)\leq \eval{\theta'} {e'}\}$ with $\theta'=\lext\sigma\circ\rho'$.

      Thus it suffices to show $L'\subseteq L$, which in turn follows from $\eval{\theta'} {e'}\leq \eval\theta e$.
      
      To prove this latter inequality, set $r\eqdef\eval\theta e$ and observe that we have $\eval\theta e=\eval{\theta_r}{e'}$ with $\theta_r=\theta,x\mapsto r$ by unfolding the definition and the fixpoint once.
      By monotonicity of the interpretation of $e'$ (\Cref{def:interpretation}), it suffices to prove $\theta'\leq\theta_r$. These two valuations coincide with $\theta$ on $X$, and at $x$ we have
      \begin{align*}
        \theta'(x) =\lext\sigma(L)=\bigvee_{a\in L}\ext\sigma a\leq r=\theta_r(x)\enspace
      \end{align*}
      where the inequality holds by definition of $L$.\qedhere
    \end{itemize}
  \end{proof}
  \shortnotations%
  The above lemma is a cornerstone of our development: it readily yields half of \Cref{thm:interpretation factors} as stated below, but also an unfolding rule for the standard language interpretation of fixpoints, which we need to prove the second half of \Cref{thm:interpretation factors}.
  \begin{corollary}
    \label{cor:key1}
    For all expressions $e\in\expr X$, for all valuations $\sigma \colon X \to M$ into a continuous model $M$, we have $\lext\sigma\interp e\leq\eval\sigma e$.    
  \end{corollary}
  \begin{proof}
    Choose $\rho=\eta_X$ in \Cref{thm:key1}.
  \end{proof}
  \begin{corollary}
    \label{cor:pfp}
    For all $x$-expressions $e$, we have $\sub{\interp e}x{\interp{\mu x.e}}\subseteq {\interp{\mu x.e}}$.    
  \end{corollary}
  \begin{proof}
    Unfold the fixpoint in the right-hand side and choose $\sigma=\override{\eta_X}x{\interp{\mu x.e}}$ in \Cref{cor:key1}.
  \end{proof}
  We can finally proceed with the proof of the main theorem from this section:
\end{textAtEnd}

The main result of this section is that all interpretations factor through the standard language interpretation:
\begin{theoremE}%
  \label{thm:interpretation factors}
  \prooflink
  For all valuations $\sigma \colon X \to M$ into a continuous model $M$ of $E$, and for all expressions $e$, we have 
  $\eval\sigma{e} = \lext\sigma\interp e\depE$. In other words, the following diagram commutes:
  \[\begin{tikzcd}[ampersand replacement=\&]
    \& {\expr{X}} \& \\
    {\lang X\depE} \&\& M
    \arrow["{\interp\_\depE}"', from=1-2, to=2-1]
    \arrow["{\ext\sigma}", from=1-2, to=2-3]
    \arrow["{\lext\sigma}"', from=2-1, to=2-3]
  \end{tikzcd}\]
  % In other words, for all expressions $e$, we have
  % $
  % \eval\sigma{e} = \bigvee_{a \in \interp{e}_E}\eval\sigma{a}
  % $.
\end{theoremE}%
\begin{proofE}%
  \longnotations%
  Let $M$ be a continuous model and recall that $\lext\sigma{\interp e}=\bigvee_{a \in \interp{e}}\eval\sigma{a}$ by definition.
  By \Cref{cor:key1}, it suffices to prove that for all sets $X$ of variables, for all expressions $e\in\expr X$, and for all valuations $\sigma \colon X \to M$, we have $\eval\sigma e\leq\lext\sigma\interp e$.
  We proceed by induction on $e$:
  \begin{itemize}    
    \item if $e$ is a variable $x$, then $\eval\sigma e=\sigma(x)=\lext\sigma(\set x)=\lext\sigma\interp e$;
    \item if $e=0$, then $\eval\sigma e={\perp}=\lext\sigma(\emptyset)=\lext\sigma\interp e$;
    \item if $e=e_1+e_2$, then
      \begin{align*}
        \eval\sigma {e}
        &=   \eval\sigma {e_1}+\eval\sigma {e_2}\\
        \tag{by induction hypothesis}
        &\leq \lext\sigma\interp{e_1}+\lext\sigma\interp{e_2}\\ 
        &= \lext\sigma(\interp{e_1}+\interp{e_2})
        = \lext\sigma\interp e 
      \end{align*}
    \item if $e=s(e_1,\dots,e_n)$, then
      \begin{align*}
        \eval\sigma {e}
        &= s(\eval\sigma {e_1},\dots,\eval\sigma {e_n}\\
        \tag{monotonicity and induction hypothesis}
        &\leq s(\lext\sigma\interp{e_1},\dots,\lext\sigma\interp{e_n})\\
        \tag{\Cref{lem:sigma:is:an:algebra:morphism}}        
        &= \lext\sigma (s(\interp{e_1},\dots,\interp{e_n})
        = \lext\sigma\interp e 
      \end{align*}
    \item finally, if $e=\mu x.e'$ for an $x$-expression $e'$, then let $r=\lext\sigma\interp e$.
      Since $\eval\sigma e$ is a least fixpoint, it suffices to show $\eval{\sigma_r}{e'}\leq r$ for $\sigma_r=\sigma,x\mapsto r$ to obtain $\eval\sigma e\leq \lext\sigma\interp e$ as required. We proceed as follows: 
      \begin{align*}
        \eval{\sigma_r}{e'}
        \tag{induction hypothesis on $e'$ with $\sigma_r$}
        &\leq\lext{\sigma_r}\interp{e'}\\
        \tag{\Cref{cor:language substitution commutes with join}}
        &=\lext{\sigma}\sub{\interp{e'}}x{\interp{\mu x.e'}}\\
        \tag{\Cref{cor:pfp}}
        &\leq\lext{\sigma}{\interp{\mu x.e'}}
        =r
      \end{align*}
      (Note again that the $x$-expression $e'$ is seen as an expression over $x\cup\set x$ when we apply the induction hypothesis.)
      \qedhere
  \end{itemize}
\end{proofE}%
\begin{textAtEnd}%
  As a corollary of the previous theorem, we can characterise the standard language interpretation of fixpoint expressions as least fixpoints, using language substitution:
  \begin{corollary}
    \label{cor:hip}
    For all $x$-expressions $e$, for all languages $L$, we have 
    $\sub{\interp{e}}{x}{L}=\eval{\eta_L}e$
    where $\eta_L=\override{\eta_X}x L$.
  \end{corollary}%
  \begin{corollary}
    \label{cor:hop}
    For all $x$-expressions $e$, we have $\interp {\mu x.e} = \mu L.\sub{\interp e}x L$.
  \end{corollary}
  \begin{proof}
    By definition, we have $\interp {\mu x.e}=\mu L.\eval{\eta_L} e$ for $\eta_L=\override{\eta_X}x L$. We may thus conclude with \Cref{cor:hip}.
  \end{proof}

  We can characterise the language of an expression as a set of equivalence classes under $E$, where we use $[t]_E$ to denote the equivalence class of term $t$ under $E$:
  \begin{corollary}\label{cor:equivclass}
    For all expressions $e\in\expr X$, $\interp e _E =\set{[t]_E\mid t\in \interp e_\emptyset}$.
  \end{corollary}
  \begin{proof}
  Define $\eta_E(x)\eqdef\{[x]_E\}$ and observe that $\eval{\eta_E} {\_}=\interp{\_}_E$~(\Cref{def:standard language interpretation}). 
  Applying~\Cref{thm:interpretation factors}, we obtain that $\interp e _E=\eval{\eta_E}e=\lext{\eta_E}\interp e _\emptyset=\bigcup_{t\in\interp e _\emptyset} \eval {\eta_E} t$. Thus it suffices to show $\eval{\eta_E}t =\set{[t]_E}$ for all terms $t$. This follows immediately as both $\eval {\eta_E}{\_}$ and $\set{[\_]}$ are homomorphisms from $\term{X}$ to $\lang X_E$, and they agree on $X$.
  \end{proof}
\end{textAtEnd}
\Cref{thm:interpretation factors} allows us to decompose $\eval\sigma{e}$ using the language $\interp{e}_E$.
This makes later proofs easier, since we can focus on atoms rather than expressions. 
We also obtain as a corollary the following characterisation of the inequational theory of continuous models:
\begin{corollaryE}
  \label{cor:sound:complete}
  \prooflink
  The standard language interpretation is sound and complete with respect to continuous models of $E$:
  for all expressions $e,f$, 
  \begin{equation*}
    \CL \models e \leq f \iff \interp{e}\depE \subseteq \interp{f}\depE\enspace.
  \end{equation*}
\end{corollaryE}
\begin{proofE}
  The left-to-right implication holds by definition, since $\interp{\_}$ is an interpretation in a continuous model. For the converse implication, if $\interp{e} \subseteq \interp{f}$ then for all interpretations $\sigma$ into a continuous model we have $\eval\sigma e\leq \eval\sigma f$ by \Cref{thm:interpretation factors}, whence $\CL \models e \leq f$.
\end{proofE}

\subsection{Inequations and Horn sentences}
\label{ssec:inequations}

An \emph{inequation} is a pair of expressions, written $e\leq f$.
A \emph{Horn sentence} $H\rightarrow e\leq f$ is a set of inequations $H$ called the \emph{hypotheses} together with an inequation $e\leq f$ called the \emph{conclusion}. Note that we allow Horn sentences with infinitely many hypotheses.

\begin{remark}
  In the present work it is more convenient to put the emphasis on inequations rather than on equations, but both are interderivable: an equation $e=f$ can be encoded as the conjunction of $e\leq f$ and $f\leq e$, and thanks to binary joins, an inequation $e\leq f$ can be encoded as the equation $e+f=f$.
\end{remark}

\begin{definition}\label{def:lattice models}
  Let $H$ be a set of hypotheses, and let $e,f$ be two expressions.
  A valuation $\sigma \colon X \to M$ into a complete algebra \emph{satisfies} the inequation $e\leq f$ if $\eval\sigma{e}\leq\eval\sigma{f}$; it satisfies $H$ if it satisfies all inequations in $H$.
  We write $M \models H \rightarrow e \leq f$
  if all valuations $\sigma \colon X \to M$ satisfying $H$ also satisfy $e \leq f$.
  We write 
  $\CL \models H \rightarrow e \leq f$
  if the above holds for all continuous models $M$ of $E$.
  When $H$ is empty, we write $\models e \leq f$ instead of $\models \emptyset \rightarrow e \leq f$.
\end{definition}
% The above definition immediately generalises to classes of complete algebras, but we are only interested in continuous models in the present work.

Note how the universal quantification on all valuations $\sigma$ is shared by the hypotheses $H$ and the conclusion $e\leq f$: we are interested in the universal Horn theory and when the signature contains a binary operation $\cdot$, we typically have
$\CL\models x\cdot x\leq x \rightarrow x\cdot(x\cdot x)\leq x$, but not
$\CL\models x\cdot x\leq x \rightarrow y\cdot(y\cdot y)\leq y$.

By \Cref{thm:interpretation factors} we can readily reduce the analysis of the Horn theory to the case where the left-hand side of the conclusion is a term rather than an arbitrary expression.
\begin{corollaryE}
    \label{cor:compare:expr:term1}
    \prooflink
    For all continuous models $M$ of $E$, for all hypotheses $H$, and for all expressions $e,f$, we have:
    $
    M \models H \rightarrow e \leq f \iff \forall a \in \interp{e}\depE,~ M \models H \rightarrow \underline a \leq f
    $.
\end{corollaryE}
\begin{proofE}
  By definition, $M \models H \rightarrow e \leq f$ if and only if for all valuations $\sigma \colon X \to M$ satisfying $H$, we have $\eval\sigma{e} \leq \eval\sigma{f}$.
  For all $\sigma$ we have $\eval\sigma{e} = \bigvee_{a \in \interp{e}} \eval\sigma{a}$ by \Cref{thm:interpretation factors}, which is equivalent to $\forall a \in \interp{e}$, $\eval\sigma{a} \leq \eval\sigma{f}$,
  which in turn is equivalent to $\forall a \in \interp{e}$, $\eval\sigma{\underline a} \leq \eval\sigma{f}$ by \Cref{lem:interpretation:term:atom}.
\end{proofE}

\section{Closed language interpretation}
\label{sec:closure}
\pratendSetLocal{category=closure}

We fix in the next two sections a set $H$ of hypotheses. We define an operator $H^* \colon \lang X_E\to \lang X_E$ which we use to alter the standard language interpretation to take $H$ into account. We then prove soundness and completeness of this new interpretation, relatively to $H$:
\begin{theorem}
  \label{thm:sound:complete}
  For all expressions $e,f$, we have:
  \begin{equation*}
    \CL \models H \rightarrow e \leq f \iff \cinterp{e}_E \subseteq \cinterp{f}_E\enspace.
  \end{equation*}
\end{theorem}
When $H$ is empty, in which case $H^*$ is the identity, we recover precisely \Cref{cor:sound:complete}.

\subsection{Contexts}

We need a notion of \emph{context} to define the operator $H^*$.
To this end, we fix a designated variable $\hole\in X$, which we call the \emph{hole}. A \emph{term context} is a term $c$ with exactly one occurrence of the hole.
Given such a term $c$, we write $cb$ for the term obtained from $c$ by replacing the hole with a term $b$; similarly, we write $ce$ for the expression obtained from $c$ by replacing the hole with an expression $e$. A \emph{context} is an atom $C$ such that $\underline C$ is a term context. This is well-defined: since $E$ is linear and regular, all terms in the equivalence class of $C$ have the same number of occurrences of a given variable. 
Given a context $C$ and an atom $a$, we write $Ca$ for the equivalence class of $\underline C\, \underline a$ (which is an atom). We extend this notation to languages: given a context $C$ and a language $L$, we write $CL$ for $\set{Ca\mid a\in L}$.

\begin{example}\label{ex:contexts}
  In monoids, where atoms are words, a context is just a pair $(l,r)$ of words; the application of such a context to a word $u$ is the composite word $lur$; its application to the language $\set{a,b}$ is the language $\set{lar,lbr}$.
  In commutative monoids, where atoms are multisets, a context is just a multiset $m$, and the application of such a context to a multiset $n$ is the multiset union $m\cup n$.
  In bimonoids, the hole of a context may appear under arbitrary alternations of sequential $(\cdot)$ and parallel $(\parallel)$ compositions. 
\end{example}

\begin{lemmaE}
  \label{lem:ctxt:std}
  \prooflink
  For all contexts $C$ and expressions $e$, we have
  $C\interp e_E = \interp{\underline C e}_E$.
\end{lemmaE}
%% APPENDIX {contextcommutes}
\begin{proofE}
We proceed by induction on $\underline C$:
  \begin{itemize}
  \item if $\underline{C} = \hole$, then $Ca = a$ for all atoms $a$, and we have
    \begin{align*}
      C \interp{e} = \set{ C a \mid a \in \interp{e} } = \set{a \mid a \in \interp{e}} = \interp{e} = \interp{\hole e} = \interp{\underline{C}e}\enspace.
    \end{align*}    
  \item if $\underline{C} = s(t_1, \dots, \underline{C}', \dots, t_n)$ for some context $C'$ and terms $t_i$, we first note that for all atoms $a$, $\interp{\underline{C}\, \underline{a}} = \set{Ca}$ (regardless of the choice of $\underline{C}$), and
  for all atoms $a$ we have,
    \begin{align*}
      \interp{\underline{C}\, \underline{a}} 
      = \interp{s(t_1, \dots, \underline{C}', \dots, t_n) \underline{a}}
      = \interp{s(t_1, \dots, \underline{C}' \underline{a}, \dots, t_n)}\enspace,
    \end{align*}
    from which we compute:
    \begin{align*}
      C \interp{e} 
      % &= \set{ C a \mid a \in \interp{e}  }\\
      = \bigcup_{a \in \interp{e}} \set{Ca}
      &= \bigcup_{a \in \interp{e}} \interp{\underline{C}\: \underline{a}}\\
      &= \bigcup_{a \in \interp{e}} \interp{s(t_1, \dots, \underline{C}' \underline{a}, \dots, t_n)}\\
      &= \bigcup_{a \in \interp{e}} s(\interp{t_1}, \dots, \interp{\underline{C}' \underline{a}}, \dots, \interp{t_n})\\
      % &= \bigcup_{a \in \interp{e}} s(\interp{t_1}, \dots, C' \interp{ \underline{a}}, \dots, \interp{t_n})\\
      &= \bigcup_{a \in \interp{e}} s(\interp{t_1}, \dots,  \set{C'a}, \dots, \interp{t_n})\\
      &= s(\interp{t_1}, \dots, C' \interp{e}, \dots, \interp{t_n})\\    
      &\tag{by induction}= s(\interp{t_1}, \dots, \interp{\underline{C}' e}, \dots, \interp{t_n})\\
      &= \interp{s(t_1, \dots, \underline{C}' e, \dots, t_n)}
      % &= \interp{s(t_1, \dots, \underline{C}', \dots, t_n)e}\\
      = \interp{\underline{C} e}\enspace.\tag*\qedhere
    \end{align*}
  \end{itemize}
\end{proofE}

\subsection{Closure under hypotheses}\label{sec:properties of closure}
In the sequel we often omit the subscript $E$ from $\interp e$ and $\lang X$ to alleviate notation. 
Recall the set $H$ of hypotheses;
we turn it into a monotone function $H\colon\lang X\to\lang X$ on languages by setting
\begin{equation*}
  % H\colon &\lang X\to \lang X\\
  % &L\mapsto
  H(L) \eqdef
  \bigcup \set{C\interp e \mid
    C \text{ a context, }
    (e\leq f)\in H,~
    C\interp f \subseteq L
  }\enspace.
\end{equation*}
\vspace{-\baselineskip}%DAMIEN: why do we get a spurious line??
\begin{definition}
  \label{def:hypothesis closure}
  A language $L$ is \emph{$H$-closed} if $H(L)\subseteq L$.
  The \emph{$H$-closure} of a language $L$ is the least $H$-closed language $H^*(L)$ containing $L$.
\end{definition}
Equivalently, $H^*$ is the least closure above the function $H$ (cf. Appendix~\ref{app:lattice}). Also note that $L$ is $H$-closed if and only if $H^*(L)=L$, and that $H^*$ is the identity function when $H$ is empty.
Below we sometimes write $H^*L$ for $H^*(L)$ and $HL$ for $H(L)$.
%  to reduce notational clutter.
\begin{example}
  \label{ex:closure}
 For monoids, if $H = \set{x \leq y+z}$ and $L = \set{y, ayb, azb}$,
  then $HL = \set{axb}$ and $H^*L=L\cup \set{axb}$; if $H = \set{xx \leq x}$ and $L = \set{xx, xyx}$,
  then $HL = \set{xxx,xxyx,xyxx}$ and $H^*L = \set{x^i, x^j y x^k \mid i \geq 2, j,k \geq 1}$.
  For commutative monoids, if $H = \set{c \leq xy, a\leq cd}$ and $L = \set{ax,by,dxy}$\footnote{We write finite multisets as words where the ordering of the letters is irrelevant.}, then $HL = \set{dc}=\set{cd}$, and $H^*L=L\cup\set{a,cd}$.
  Finally, in bimonoids, if $H=\set{b\leq a\parallel c}$ and $L=\set{x \cdot (a \parallel b\parallel c)}$, then $H^*L=L\cup\set{x\cdot(b\parallel b)}$.
\end{example}
The main point of $H$-closed languages is that they validate the inequations in $H$:
\begin{lemmaE}
  \label{lem:cinterp:hyps}
  \prooflink
  For all $(e\leq f)\in H$, we have $\cinterp e\subseteq \cinterp f$.
\end{lemmaE}
\begin{proofE}
  This follows from $\interp e\subseteq H(\interp f)$, which is immediate from the definition of the function associated to $H$ by choosing the identity context for $C$.
\end{proofE}

\subsection{Soundness and completeness}
\label{ssec:sound:complete}

To prove \Cref{thm:sound:complete}, we first show that $\cinterp{\_}$ is an interpretation into a continuous model. We write $\clang X$ for the set of $H$-closed languages, ordered by inclusion. In general, the function $H^*$ does \emph{not} preserve unions and is \emph{not} a homomorphism from $\lang X$ to $\lang X$. Therefore its image $\clang X$ is not a model \emph{as is}.
To get a model of closed languages, we need to adapt the operations and close them explicitly, as follows.

For all collections $(L_i)_{i\in I}$ of $H$-closed languages, and for all symbols $s$ of arity $n$, set:
\phantomsection\label{page:bigvee}
\begin{align*}
  \bigvee_{i\in I} L_i \eqdef H^*\paren{ \bigcup_{i\in I} L_i}
  &&
     \begin{aligned}
       s_{\clang X} \colon &(\clang X)^n \to \clang X \\
       &(L_1, \dots, L_n) \mapsto H^*(s_{\lang X}(L_1, \dots, L_n))
     \end{aligned}
     \enspace
\end{align*}
\begin{textAtEnd}%
  We establish a series of lemmas that lead to \Cref{lem:closed:continuous:model}: the set $\clang X$ of $H$-closed languages, equipped with the explicitly closed operations, is a continuous model.

  By general properties of closure operators, we have:
  \begin{lemma}
    \label{lem:closure commutes with join}
    Given a collection $(L_i)_{i\in I}$ of languages, we have
    \begin{align*}
      H^*\bigcup_{i\in I}L_i = H^*\bigcup_{i\in I}H^*L_i\enspace.
    \end{align*}
  \end{lemma}
  \begin{proof}
    The left-in-right inclusion follows from monotonicity of $H^*$ and the fact that $L\subseteq H^*L$.
    For the converse inclusion, we have $\bigcup_{i\in I}H^*L_i\subseteq H^*\bigcup_{i\in I}L_i$, again by monotonicity on $H^*$, hence the expected inclusion by applying $H^*$ on both sides.
  \end{proof}
  Moreover, the operation $\bigvee$ as defined on page~\pageref{page:bigvee}, is indeed a supremum operation:
  \begin{lemma}
    $(\clang X,\bigvee)$ is a complete lattice.
  \end{lemma}
  \begin{proof}
    For all subsets $\set{L_i \mid i \in I} \subseteq \clang X$, $\bigvee_{i \in I} L_i$ contains $\bigcup_{i \in I} L_i$ and is thus an upper bound of $\set{L_i \mid i \in I}$; it remains to show that this is the least one in $\clang X$.
    Let $M \in \clang X$ such that $L_i \subseteq M$ for all $i \in I$;
    we have $\bigcup_{i \in I}L_i \subseteq M$, hence $H^*\bigcup_{i \in I}L_i \subseteq H^*M = M$ since $H^*$ is monotone and $M$ is $H$-closed.
    % Thus $\bigvee_{i \in I}L_i = H^*( \bigcup_{i \in I}L_i )$ is indeed a least upper bound.
  \end{proof}

  We have defined two families of operations on languages: $s_{\lang X}$ on arbitrary languages, and
  $s_{H^*\lang X}=H^*\circ s_{\lang X}$ on $H$-closed languages. For the sake of readability, we often omit the subscripts as before, but we never use $s$ to mean $s_{H^*\lang X}$, writing $H^*s$ instead in that case.
  
  In the next two proofs, we use some tools from~\cite[Appendix~A]{pous2024tools}, which we adapted to our needs in Appendix~\ref{app:lattice}. 
  \begin{lemma}\label{lem:closure commutes with distributive symbol}
    For all symbols $s$ of arity $n$, for all languages $L_i \in \lang X$, we have: 
    \begin{align*}
      s(L_1, \dots, H^*L_i, \dots, L_n) \subseteq H^*s (L_1,\dots, L_n)\enspace.
    \end{align*}
  \end{lemma}
  \begin{proof}
    % Using the terminology from \Cref{app:lattice}, we have to show that $H^*$ is contextual w.r.t. $s_{\lang X}$. 
    By \Cref{lem:contextual implies fixed point contextual}, since $s$ is continuous, it suffices to show % that $H$ is contextual w.r.t.\ $s$, i.e.,
    \begin{align*}
      s(L_1, \dots, HL_i, \dots, L_n) \subseteq Hs (L_1,\dots, L_n)\enspace.
    \end{align*}
    Let $a$ be an atom in the left-hand side. By definition $a=s(a_1,\dots,Cb,\dots,a_n)$ with each $a_j$ in $L_j$, $C$ a context, and $b\in\interp e$ for some inequation $(e\leq f)\in H$ with $C\interp f\subseteq L_i$.
    For $C'=s(a_1,\dots,C,\dots,a_n)$, we check easily that $C'\interp f\subseteq s(L_1,\dots,L_n)$, so that $C'\interp e \subseteq Hs(L_1,\dots,L_n)$. As $a=C'b\in C'\interp e$, we can conclude that $a$ belongs to the right-hand side.
  \end{proof}

  \begin{corollary}
    \label{lem:closurefullycontextual}
    For all contexts $C$ and languages $L$, we have $CH^*L\subseteq H^*CL$. 
  \end{corollary}
  \begin{proof}
    Generalising \Cref{lem:closure commutes with distributive symbol} by induction on the context $C$.
  \end{proof}

  \begin{corollary}
    \label{lem:closureishomomorphism}
    The function $H^*$ is a homomorphism from $\lang X$ to $\clang X$:
    for all symbols $s$ of arity $n$ and for all languages $L_i \in \lang X$, we have
    \begin{align*}
      H^*s(L_1, \dots, L_n) = H^*s (H^*L_1,\dots, H^*L_n)\enspace.
    \end{align*}    
  \end{corollary}
  \begin{proof}
    The left-to-right inclusion holds by monotonicity and $L\subseteq H^*L$.
    The right-to-left inclusion follows from \Cref{lem:closure commutes with distributive symbol} by \Cref{lem:contextual characterization for closures}.
  \end{proof}
\end{textAtEnd}%
\begin{propositionE}
  \label{lem:closed:continuous:model}
  \prooflink
  The set $\clang X$ of $H$-closed languages, equipped with the above operations, is a continuous model of $E$.
\end{propositionE}
\begin{proofE}
  By \Cref{lem:closureishomomorphism}, $\clang X$ is the homomorphic image of $\lang X$ under $H^*$ (this function is surjective). Since $\lang X$ satisfies $E$ (\Cref{thm:linreg:pset}), so does $\clang X$.

  It remains to show that its operations preserve arbitrary joins.
  Let $s$ be a symbol of arity $n+1$, let $L_1\dots L_n \in \clang X$, and let $\set{K_i \mid i \in I} \subseteq \clang X$; we have:
  \begin{align*}
    &H^*s(L_1, \dots, \bigvee_{i\in I} K_i, \dots, L_n)\\
    =~&H^*s(L_1, \dots, H^*(\bigcup_{i\in I} K_i), \dots, L_n)\\
    \tag{\Cref{lem:closure commutes with distributive symbol}}        
    \subseteq~&H^*s(L_1, \dots, \bigcup_{i\in I} K_i, \dots, L_n)\\
    \tag{$s$ is continuous on $\lang X$}
    =~& H^* \bigcup_{i \in I} s(L_1, \dots, K_i, \dots, L_n)\\
    \tag{\Cref{lem:closure commutes with join}}
    =~& H^* \bigcup_{i \in I} H^*s(L_1, \dots, K_i, \dots, L_n)\\
    =~& \bigvee_{i \in I} H^*s(L_1, \dots, K_i, \dots, L_n)\\
    \tag*{(monotonicity)~\qedhere}
      \subseteq~&H^*s(L_1, \dots, \bigvee_{i\in I} K_i, \dots, L_n)\enspace.
  \end{align*}
\end{proofE}
It remains to show that $\cinterp{\_}$ actually is an interpretation into that model:

\begin{propositionE}
  \label{lem:interpretation agrees on expressions with closure}
  \prooflink
  Let $\eta_H \colon X \to \clang X$ be the valuation defined by $\eta_H(x)\eqdef H^*\set x$.
  For all expressions $e$ we have $\cinterp{e} = \eval{\eta_H}{e}$.
\end{propositionE}
\begin{proofE}
  We first observe that for all atoms $a$, we have
  $\eval{\eta_H}{a}= H^*\set a$.
  Indeed, both $\ext{\eta_H}$ and $H^*\set\_$ are homomorphisms (\Cref{lem:closureishomomorphism}), and they agree on variables.
  Then we compute:
  \begin{align*}
    \tag{\Cref{thm:interpretation factors}}
    \eval{\eta_H}{e}
    = \bigvee_{a \in \interp{e}} \eval{\eta_H}{a}
    &= H^*\bigcup_{a \in \interp{e}} \eval{\eta_H}{a}\\
    \tag{preliminary observation}
    &= H^*\bigcup_{a \in \interp{e}} H^*\set{a}\\
    \tag{\Cref{lem:closure commutes with join}}
    &= H^*\bigcup_{a \in \interp{e}} \set{a}\\
    &= \cinterp{e}\enspace. \tag*{\qedhere}
  \end{align*}
\end{proofE}
We can finally prove our main theorem.
\begin{proof}[Proof of \Cref{thm:sound:complete}]
  The left to right implication follows directly from \Cref{{lem:interpretation agrees on expressions with closure}} and the fact that our interpretation satisfies all hypotheses in $H$ (\Cref{lem:cinterp:hyps}).

  For the converse implication, assume $\cinterp{e} \subseteq \cinterp{f}$ and pick a continuous model $M$; we have to show $M\models H \rightarrow e \leq f$.
  Define the following language:
  \begin{align*}
    F \eqdef \{ a \in \atom{X} \mid M \models H \rightarrow \underline a \leq f \}\enspace.
  \end{align*}
  By \Cref{cor:compare:expr:term1}, we have that for all expressions $g$,
  \begin{align*}
    \tag{$\star$}
    \interp g \subseteq F \iff M\models H\rightarrow g \leq f\enspace.
  \end{align*}
  Thus it suffices to show $\interp e \subseteq F$.
  By assumption, this follows from $\cinterp f \subseteq F$, which
  in turn follows from $\interp{f} \subseteq F$ and $H(F) \subseteq F$, by definition of $H^*$.
  We get $\interp{f} \subseteq F$ directly from $(\star)$, and it
  only remains to show $H(F) \subseteq F$.

  Fix a context $C$, a hypothesis $(e_0\leq f_0)\in H$, and suppose $C\interp{f_0}\subseteq F~(\dagger)$.
  We have to show $C\interp{e_0}\subseteq F$.
  From \Cref{lem:ctxt:std} and $(\star)$, we have that for all expressions $g$,
  \begin{align*}
    \tag{$\star\star$}
    C\interp g\subseteq F \iff (\forall\sigma%\colon X\to M,~\eval\sigma{H}\implies
    \text{ satisfying $H$ in }M,~ \eval\sigma{\underline C\,g}\leq\eval\sigma{f})\enspace.
  \end{align*}
  Accordingly, fix a valuation $\sigma\colon X\to M$ satisfying $H$.
  Since we have $(e_0\leq f_0)\in H$, we know $\eval\sigma{e_0}\leq\eval\sigma{f_0}$.
  We derive
    $\eval\sigma{\underline C e_0}
    \leq~\eval\sigma{\underline C f_0}
    \leq~\eval\sigma{f}$ first by monotonicity and then using $(\dagger)$ with $(\star\star)$. 
  We conclude using $(\star\star)$ again.
\end{proof}

\section{Axiomatisations}
\label{sec:axiomatisation}
\pratendSetLocal{category=axiomatisation}

In this section we study axiomatisations. We restrict to quasi-equational ones, which consist only of Horn sentences. In order to deal with standard examples from the literature, we study those axiomatisations on arbitrary fragments of our general syntax of expressions.

\begin{definition}
  \label{def:fragment}
  A \emph{fragment} is a subset $\F$ of expressions containing $0$ and the variables, closed under $+$, the operations from $\Sigma$, application of substitutions $\theta\colon X\pto \F$, (closed) sub-expressions, and such that $\sub e x f$ belongs to $\F$ whenever $\mu x.e$ and $f$ belong to $\F$.
\end{definition}
\begin{example}
  The set of all (closed) expressions and the set of fixpoint-free expressions is are fragments.
  On monoids, the set of expressions where all fixpoint sub-expressions are of the form $\mu x.1+e\cdot x$ is a fragment, which corresponds to regular expressions.
  On bimonoids, the set of expressions where all fixpoint sub-expressions are either of the form $\mu x.1+e\cdot x$ or of the form $\mu x.1+e\parallel x$ is a fragment, which corresponds precisely to bi-KA expressions~\cite{struth2014bikleene}. 
\end{example}

\begin{definition}
  A \emph{(quasi-equational) axiomatisation over a fragment} is a set of Horn sentences expressed in that fragment. Given such an axiomatisation $Q$, its \emph{Horn theory} is the judgement $\derive Q H e\leq f$ inductively defined by the rules in \Cref{fig:rules}, where $H$ is a set of hypotheses in the fragment. 
\end{definition}

\begin{figure}
  \begin{mathpar}
    \inferrule
    {\forall (e'\leq f')\in H',~\derive Q H e'\theta\leq f'\theta}
    {\derive Q H e\theta\leq f\theta}
    {(H'\to e\leq f)\in Q,~\theta\colon X\pto\F}
    \\ %
    \inferrule{ }
    {\derive Q H e\leq f}
    {(e\leq f)\in H}
    \\ %
    \inferrule{ }
    {\derive Q H e\leq e}
    {e\in\F}
    \and %
    \inferrule{\derive Q H e\leq f \and \derive Q H f\leq g}
    {\derive Q H e\leq g}
    \\ %
    \inferrule{\derive Q H e_1\leq f_1 \and \derive Q H e_2\leq f_2}
    {\derive Q H e_1+e_2\leq f_1+f_2}
    \and %
    \inferrule{\forall i,~\derive Q H e_i\leq f_i}
    {\derive Q H s(e_1,\dots,e_n)\leq s(f_1,\dots,f_n)}
    {s^{(n)}\in \Sigma}    
  \end{mathpar}
  \caption{Deduction rules for the Horn theory of an axiomatisation $Q$ over a fragment $\F$.}
  \label{fig:rules}
\end{figure}

Regarding variables and quantifications, such a judgement should intuitively be understood as the following implication:
\begin{center}
  ``$(\forall \vec x,~Q) \Rightarrow \forall \vec x,~(H \Rightarrow e\leq f)$''\enspace.
\end{center}
This is why we use a substitution in the first rule of \Cref{fig:rules}, which makes it possible to use any instance of some axiom in $Q$, but not in in the second one, which only makes it possible to use an hypothesis from $H$, \emph{as is}. The next two rules enforce reflexivity and transitivity, while the last two rules ensure monotonicity of $+$ and the operations from $\Sigma$.

We establish the basic properties of these deduction rules in Appendix~\ref{app:axiomatisation}.
\begin{textAtEnd}%
  Before proving \Cref{prop:rules:sound},
  we observe that the interpretation of a (closed) expression does not depend on the values assigned to recursion variables, and we show how our interpretation function (\Cref{def:interpretation}) distributes over substitutions (page~\pageref{page:subst}).
  \longnotations
  \begin{fact}
    \label{fact:norecdep}
    For all valuations $\rho,\rho' \colon X\cup R \to A$ into a complete algebra $A$, if $\rho$ and $\rho'$ agree on $X$, then for all expressions $e\in \expr X$, we have $\evalb{\rho}e=\evalb{\rho'}e$.
  \end{fact}
  \begin{lemma}
    \label{lem:substitution commutes with interpretation}
    For all sets of variables $X$, possibly open expressions $e$, substitutions $\rho\colon X\cup R\pto\expr{X}$, and valuations $\rho' \colon X\cup R \to A$ into a complete algebra $A$, we have:
    \begin{align*}
      \evalb{\rho'}{e\rho} = \evalb{\rho''}{e}\text{ where }
      \rho''(x) =
      \begin{cases}
        \evalb{\rho'}{\rho(x)} &\text{if $\rho(x)$ is defined,}\\
        \rho'(x) &\text{otherwise.}
      \end{cases}
   \end{align*}
  \end{lemma}
  \begin{proof}
    We proceed by induction on the expression $e$. In the case $s(e_1, \dots, e_n)$, we have
    \begin{align*}
      \evalb{\rho'}{(s(e_1, \dots, e_n))\rho}
      &= \evalb{\rho'}{s(e_1\rho , \dots, e_n\rho)} \\% & (\text{\Cref{def:expressions substitution}})\\
      &= s(\evalb{\rho'}{e_1\rho}, \dots, \evalb{\rho'}{e_n\rho})\\%  & (\text{\Cref{def:interpretation}})\\
      \tag{by induction}
      &= s(\evalb{\rho''}{e_1}, \dots, \evalb{\rho''}{e_n})\\%  & (\text{I.H.})\\
      &= \evalb{\rho''}{s(e_1, \dots,e_n)}\enspace.%  & (\text{\Cref{def:interpretation}}).
    \end{align*}
    The case $0$ is trivial and the case $e_1+e_2$ is similar to the previous one.
    For the variable case $x$ with $x \in X \cup R$, we have
    $\evalb{\rho'}{x\rho} = \evalb{\rho'}{\rho(x)} = \rho''(x) =
    \evalb{\rho''}{x}$ if $\rho(x)$ is defined, and $\evalb{\rho'}{x\rho} = \evalb{\rho'}{x} = \rho'(x) =
    \evalb{\rho''}{x}$ otherwise.
    It remains to deal with fixpoint expressions $\mu x.e$, for which we compute:
    \begin{align*}
      \evalb{\rho'}{(\mu x.e)\rho}
      &= \evalb{\rho'}{\mu x.(e (\rho\backslash x))} \\% & (\text{\Cref{def:expressions substitution}}, x \neq x')\\
      &= \mu a.\evalb{\override{\rho'}{x}{a}}{e (\rho\backslash x)}\\%  & (\text{\Cref{def:interpretation}})\\
      &= \mu a.\evalb{\override{\rho''}{x}{a}}{e}\\% & (\text{I.H.})\\
      &= \evalb{\rho''}{\mu x.e}\enspace.
    \end{align*}
    The last-but-one step deserves further explanation.    
    We use the induction hypothesis on $e$, using the substitution $\rho\backslash x$ and the valuation $\override{\rho'}{x}{a}$, which gives $\evalb{\override{\rho'}{x}{a}}{e (\rho\backslash x)}=\evalb{\rho'''}{e}$ where $\rho'''(y)=\evalb{\override{\rho'}{x}{a}}{(\rho\backslash x)(y)}$ if $(\rho\backslash x)(y)$ is defined, and $\rho'''(y)=(\override{\rho'}{x}{a})(y)$ otherwise. We observe that 
    $\rho'''=\override{\rho''}{x}{a}$ with $\rho''$ defined as in the statement (using \Cref{fact:norecdep} when the argument is in the domain of $\rho\backslash x$).
  \end{proof}

  In particular, for  all valuations $\sigma\colon X\to A$, we get that:
  \begin{itemize}
  \item for all expressions $e$ and all substitutions $\theta\colon X\pto\expr{X}$,
    \begin{align*}
      \eval\sigma{e\theta} = \eval{\sigma_\theta}{e}\text{ where }
      \sigma_\theta(x) = 
      \begin{cases}
        \eval\sigma{\theta(x)} &\text{if $\theta(x)$ is defined,}\\
        \sigma(x) &\text{otherwise;}
      \end{cases}
    \end{align*}
  \item and for all $x$-expressions $e$, $\eval\sigma{\mu x.e} = \ext{\sigma}_{x \mapsto \mu x.e}(e) = \eval\sigma{\sub e x {\mu x.e}}$.
  \end{itemize}
  \shortnotations
  Moreover, combined with~\ref{thm:interpretation factors}, the first point yields:
  \begin{corollary}
    \label{cor:interp:subst}
    Let $e$ be an expression and let $\theta\colon X\pto\expr X$ be a substitution. We have
    $\interp{e\theta}=\lext{\sigma}\interp e$ with $\sigma(x)\eqdef\interp{\theta(x)}$ if $\theta(x)$ is defined and $\sigma(x)\eqdef\set x$ otherwise.
  \end{corollary}
\end{textAtEnd}%
Since they preserve validity in all complete models, we get:
\begin{theoremE}
  \label{prop:rules:sound}
  \prooflink
  Let $Q$ be an axiomatisation and let $M$ be a complete model.
  If $M\models Q$ and $\derive Q H e \leq f$ then $M\models H \rightarrow e \leq f$.
\end{theoremE}
\begin{proofE}
  By induction on the derivation; the only difficulty is with the first rule from \Cref{fig:rules}, which involves substitutions and for which we use \Cref{lem:substitution commutes with interpretation}.
\end{proofE}

We provide a generic axiomatisation in \Cref{fig:naive}, which we call \emph{naive}.
Given a fragment $\F$, this axiomatisation $N(E,\F)$ comprises axioms ensuring that we get a semilattice with $+$ and $0$, that all operations from $\Sigma$ distribute over finite joins, that the equations from $E$ are satisfied, and that fixpoint expressions that are allowed in $\F$ are indeed least (pre)fixpoints.

\begin{figure}
  \begin{mathpar}
    \begin{array}{r@{~}l}
      0&\leq x\\
      x+x&\leq x\\
      x&\leq x+y\\
      y&\leq x+y\\
    \end{array}
    \and
    \begin{array}{r@{~}lc}
      s(\vec x,0,\vec z)&\leq 0& s^{(n+1)}\in\Sigma\\
      s(\vec x,y+y',\vec z)&\leq s(\vec x,y,\vec z)+s(\vec x,y',\vec z)& s^{(n+1)}\in\Sigma\\
      u&\leq v &(u=v) \in E\\
      v&\leq u &(u=v) \in E\\
    \end{array}\\
    \begin{array}{c@{\qquad}c}
      \sub e x {\mu x.e}\leq \mu x.e& \mu x.e\in\F\\
      \sub e x f\leq f\rightarrow \mu x.e\leq f& \mu x.e, f\in\F\\
    \end{array}
  \end{mathpar}
  \caption{Naive axiomatisation $N(E,\F)$ on equations $E$ over a fragment $\F$.}
  \label{fig:naive}
\end{figure}

As expected, this naive axiomatisation is sound for all continuous models:
\begin{propositionE}
  \label{lem:naive:sound}
  \prooflink
  For all fragments $\F$, we have $\CL\models N(E,\F)$.
\end{propositionE}
\begin{proofE}
  Again, the only difficulty comes from the substitutions involved in the last two axioms, about fixpoints, for which we use \Cref{lem:substitution commutes with interpretation}.
\end{proofE}
Combined with \Cref{prop:rules:sound}, we get that all Horn sentences derivable from $N(E,\F)$ are valid in all continuous models of $E$.

In the case of trees, i.e., when $E$ is empty, and when considering all expressions $(\F=\expr{X})$, Ésik has proved that the naive axiomatisation is complete for the empty set of hypotheses:
\begin{theorem}[{\cite[Theorem~21 with Remark~27]{esik2010tree}}]
  \label{thm:completeness for no equations}
  For all expressions $e,f$ we have
  \begin{align*}
    N(\emptyset,\expr{X}) \vdash e \leq f \iff \interp{e}_\emptyset \subseteq \interp{f}_\emptyset \iff \C\emptyset\models e\leq f\enspace.
  \end{align*}
\end{theorem}
Unfortunately, this result does not generalise to arbitrary equations $E$ (and even less to arbitrary hypotheses $H$, as those can be used to emulate equations).

\begin{example}\label{example:countermodel}
  Consider for instance monoids and regular expressions, hence Kleene algebra.
  A counter-model can be used to show that the naive axiomatisation cannot derive $a^*\cdot a \leq a^*$ (see~\Cref{countermodel}, where we provide a complete model of the naive axiomatisation which is not continuous and refutes this law, although $\interp{a^*\cdot a}\subseteq \interp{a^*}$).
  Instead, several axiomatisations have been proposed by Conway~\cite{conway1971regular}, which were later proved complete by Krob~\cite{Krob91a}, Kozen~\cite{Kozen91} and Boffa~\cite{boffa1990remarque,Boffa95}.
  Amongst them we find \emph{left-handed Kleene algebra}~\cite{KozenS12,ddp:lpar18:lefthanded,Kappe23}, which is close to the naive axiomatisation, except that the least fixpoint induction
  axiom from the naive axiomatisation, which can be reformulated here as $1+e\cdot f\leq f \rightarrow e^*\leq f$,
  % \begin{align*}
  %   1+a\cdot f\leq f \rightarrow a^*\leq f\enspace,
  % \end{align*}
  is strengthened into $h+e\cdot f\leq f \rightarrow e^*\cdot h\leq f$.
  % \begin{align*}
  %   b+a\cdot f\leq f \rightarrow a^*\cdot b\leq f\enspace. 
  % \end{align*}
  The resulting ability to perform inductions on stars ``under a context $h$ on the right'' is crucial for completeness.
\end{example}

We conclude this section by showing that there are nevertheless situations where the naive axiomatisation is enough in the presence of both equations and hypotheses: when some of the expressions are fixpoint-free. We fix a fragment $\F$ in the remainder of this section.

\begin{textAtEnd}%

  \begin{counterexample}
    \label{countermodel}
    We show here that the naive axiomatisation is incomplete on monoids.
    We do so by providing a model which is complete so that we have all least fixpoints, including Kleene star, but where $a^*\cdot a\leq a^*$ does not hold although we have $\interp{a^*\cdot a}\subseteq\interp{a^*}$.
    
    As elements for the model we take $0$, $a^i$ for all $i\in\mathbb N$, $a^*$ and $\top$, with the following complete linear order:
    \begin{align*}
      0< 1 < a < a^2 < \dots < a^* < \top
    \end{align*}
    The bottom element is $0$, $+$ is given by max, $1$ is $a^0$ and we define product as follows:
    \begin{itemize}
    \item $0\cdot x= x\cdot 0 =0$,
    \item $1\cdot x = x\cdot 1 =x$,
    \item $a^i\cdot a^j = a^{i+j}$,
    \item $a^i\cdot a^*=a^*$,
    \item and $x\cdot y = \top$ in the remaining cases (note in particular that $a^*\cdot a^i=\top$ whenever $i>0$).
    \end{itemize}
    This is a partially ordered monoid, and all least fixpoints of monotone functions are available since the order is complete.
    The naive axioms from \Cref{fig:naive} are all satisfied, including distributivity of product over finite suprema.
    
    We may moreover compute Kleene star $(x^*=\mu y.1+x\cdot y)$ and observe that we have:
    \begin{itemize}
    \item $0^*=1^*=1$,
    \item $(a^i)^*=a^*$ for $i>0$,
    \item $(a^*)^* = \top^* =\top$.
    \end{itemize}
    Thus we have $a^*\cdot a = \top \not\leq a^*$. 
    The trick is that product, while distributing over finite suprema, is not continuous in its second argument: this is not a continuous model.        
  \end{counterexample}
  
  We continue this appendix by stating four basic properties of derivations; their proofs are routine and thus omitted.
  
  In the first rule and in the reflexivity rule, we make sure that we only construct expressions in the fragment $\F$ so that we have:
  \begin{lemma}
    If $Q$ is an axiomatisation over a fragment $\F$ and $\derive Q H e\leq f$, then both $e$ and $f$ belong to $\F$.
  \end{lemma}
  
  Derivations are stable under substitutions, as long as substituted variables are not involved in the hypotheses:
  \begin{lemma}
    If $Q$ is an axiomatisation over a fragment $\F$ and $\derive Q H e\leq f$, then $\derive Q H e\theta\leq f\theta$ for all substitutions $\theta\colon X\pto\F$ whose domain does not contain variables appearing in $H$.
  \end{lemma}
    
  The following lemma makes it possible to transport proofs from one axiomatisation to another.
  There we write $\derive Q H H'$ when all inequations in $H'$ can be derived from $Q$ and $H$, and $Q \proves Q'$ when all sentences from $Q'$ can be derived from $Q$.
  \begin{lemma}
    \label{lem:derivation strengthening}
    Let $Q, Q'$ be axiomatisations over a given fragment and
    let $H, H'$ be sets of inequations in that same fragment.
    \begin{align*}
      \text{If }
      \begin{cases}
        Q \proves Q' \\
        \derive Q H H'
      \end{cases}
      \!\!
      \text{then for all expressions }e,f,~
      \derive {Q'} {H'} e \leq f
      \text{ implies }
      \derive Q H e \leq f\,.
    \end{align*}
  \end{lemma}
  Note that the first two requirements are always satisfied when $Q'\subseteq Q$ and $H'\subseteq H$, which is a special case we frequently face.
  
  \begin{lemma}
    \label{lem:extend:id}
    Let $Q$ be an axiomatisation over a fragment $\F$ deriving all the sentences in $\Naivegen$ and
    let $H$ be a set of hypotheses over $\F$.
    If $\theta\colon X\pto \F$ is a substitution such that for all variables $x$ in its domain we have $\derive Q H \theta(x) = x$, then
    for all expressions $e\in\F$, we have $\derive Q H e\theta = e$.
  \end{lemma}
  
  \medskip

  Now we move to the proof of \Cref{lem:expressions can be compared using terms}.
  First we provide a concrete description of language substitution on term languages, that is when $E$ is empty.
  \longnotations%
  \begin{lemma}
    \label{lem:explicit:term:lsubst}
    Let $X,Y$ be sets of variables and let $x$ be a variable not in $X$.
    Let $K\subseteq\term{X\cup\set x}$ and $L\subseteq\term{Y}$ be two term languages.
    We have
    \begin{align*}
      \sub K x L = \set{t\sigma \mid Z\text{ a set },~t\in \term{X\uplus Z},~\sigma\colon Z \to L,~t(Z \mapsto x)\in K}\enspace.
    \end{align*}
    (There we write $\sigma \colon Z \to L$ to mean that $\sigma$ is a substitution mapping each variable in $Z$ to a term in $L$, seen as an expression.)
  \end{lemma}
  \begin{proof}[Proof sketch]
    Observe that $\sub K x L$ consists of terms in $\term{X\cup Y}$, which are obtained from terms in $K$ by replacing each occurrence of $x$ with a term in $L$.
    The set $Z$ makes it possible to distinguish the occurrences of $x$, the substitution $\sigma$ sends those occurrences to terms in $L$, and the substitution $Z\mapsto x$ maps all of them to $x$ to recover the term from $K$.
  \end{proof}
  \shortnotations%     
\end{textAtEnd}%

First we show that the naive axiomatisation proves memberships: %
\begin{propositionE}
  \label{lem:expressions can be compared using terms}
  \prooflink
  For all expressions $e\in \F$ and for all $a\in\interp e_E$, we have
  $\Naivegen\vdash \underline a \leq e$.
\end{propositionE}
\begin{proofE}
  \longnotations
  Without loss of generality, we may assume that $E$ is empty.
  Indeed, if an atom $a$ belongs to $\interp e_E$ then $\underline a$ is equivalent via $E$ to a term $t\in \interp e_\emptyset$ via~\Cref{cor:equivclass}; since $\Naivegen$ contains $E$, the statement with $E=\emptyset$ is sufficient to conclude by transitivity.

  In order to be able to proceed by induction while staying in the fragment $\F$ (whose variables are in a fixed $X$), we prove the following generalisation:
  \renewcommand\Naivegen{N(\emptyset,\F)}%
  \begin{quote}
    For all sets $Y$, for all expressions $e\in \expr Y$, for all substitutions $\theta\colon Y\pto \F$ such that $e\theta\in \F$, and for all $t\in \interp e_\emptyset$, we have 
    $\Naivegen\vdash t\theta \leq e\theta$.
  \end{quote}
  In the remainder of this proof we just write $\interp e$ for $\interp e_\emptyset$ and $N$ for $\Naivegen$.

  We proceed by induction on $e$
  \begin{itemize}    
  \item if $e$ is a variable $x$ then we conclude by reflexivity;
  \item if $e$ is $0$ then $\interp e = \emptyset$ and the statement vacuously holds.
  \item if $e=e_1+e_2$, then $e\theta\in\F$ implies that both $e_1\theta$ and $e_2\theta$ belong to $\F$ since a fragment must be closed under subexpressions, and we can easily conclude using the induction hypothesis and the axioms about $+$;
  \item if $e=s(e_1,\dots,e_n)$, then $e\theta\in\F$ similarly implies that $e_i\theta$ belongs to $\F$ for all $i$, and we can easily conclude using the induction hypothesis and the monotonicity rule for $s$;
  \item if $e=\mu x.e'$ for some $x$-expression $e'$ then consider the language
    $L\eqdef \set{t \mid N\vdash t\theta \leq e\theta}$, so that we have to prove  $\interp e\subseteq L$. 
    By~\Cref{cor:hop} it suffices to show $\sub{\interp{e'}}{x}{L} \subseteq L$.
    
    Let $t' \in \sub{\interp{e'}}{x}{L}$. By \Cref{lem:explicit:term:lsubst}, we have $t'=t\sigma$ for some set $Z$,
    term $t\in \term{Y\uplus Z}$, and substitution $\sigma\colon Z \to L$ such that $t(Z \mapsto x)\in \interp{e'}$.
    We derive:
    \begin{align*}
      \tag{with $\rho(z)\eqdef \sigma(z)\theta$}
      N \vdash      
      t'\theta = t\sigma\theta
      &=t\theta\rho\\
      \tag{monotonicity with $\forall z,~\sigma(z)\in L$}
      &\leq t\theta(Z \mapsto e\theta)\\
      \tag{with $\theta'\eqdef\theta,x\mapsto e\theta$}
      &=t(Z \mapsto x)\theta'\\
      \tag{induction hypothesis on $e'$}
      &\leq e'\theta'\\
      &=\sub{e'}xe\theta\\
      \tag{fixpoint unfolding axiom}
      &\leq e\theta
    \end{align*}
    thus $t'\in L$, as required.
    (When we use the induction hypothesis, we have to check that $e'\theta'$ belongs to $\F$, which is the case because $e\theta$ does, thanks to our last requirement on fragments---\Cref{def:fragment}).
    \qedhere
  \end{itemize}
\end{proofE}
We deduce that fixpoint-free expressions can be rewritten as finite sums of terms:
\begin{lemmaE}
  \label{lem:fixpointfree:finite}
  \prooflink
  For all fixpoint-free expressions $e$, $\interp e_E$ is finite and 
  $\Naivegen\vdash e=\sum_{a\in\interp e}\underline a$.
\end{lemmaE}%
\begin{proofE}%
  By induction on $e$, using continuity in the case of function symbols from $\Sigma$.
\end{proofE}%
\begin{textAtEnd}%
  Thus we may reason pointwise to prove upper-bounds of fixpoint-free expressions:
  \begin{corollary}
    \label{cor:fixpointfree:prove}
    For all hypotheses $H$, for all fixpoint-free expressions $e$, and for all expressions $f$,
    $\derive\Naivegen H e\leq f$ iff $\forall a\in\interp e_E,~\derive\Naivegen H\underline a\leq f$.
  \end{corollary}
\end{textAtEnd}%
It follows that the naive axiomatisation is complete for inequations whose left-hand side is fixpoint-free, when the right-hand sides of the hypotheses are also fixpoint-free.
\begin{theoremE}
  \label{thm:partial completeness for axioms}
  \prooflink
  Let $H = \set{e_i \leq f_i \mid i\in I}$ be a set of hypotheses and let $e,f$ be expressions.
  If $e$ and all $f_i$ are fixpoint-free, then we have
  \begin{align*}
    \derive\Naivegen H e \leq f \iff \cinterp{e}_E \subseteq \cinterp{f}_E\enspace.
  \end{align*}
\end{theoremE}
\begin{proofE}
    The left-to-right implication follows from \Cref{prop:rules:sound,lem:naive:sound}.
    For the converse implication, abbreviate $\Naivegen$ as $N$ and consider the following language:
    \begin{align*}
      L \eqdef \{ a \in \atom X \mid \derive N H \underline a \leq f \}\enspace.
    \end{align*}
    Below we show $\interp{f} \subseteq L$ and $HL \subseteq L$.
    From this we deduce $\cinterp{f} \subseteq L$, whence $\interp{e} \subseteq L$ by assumption, and finally $\derive N H e \leq f$ by \Cref{cor:fixpointfree:prove}, since $e$ is fixpoint-free.

    \Cref{lem:expressions can be compared using terms} immediately yields $\interp{f} \subseteq L$.  
    To show $HL \subseteq L$, let $a \in HL$.
    We have $a \in \interp{Ce_i}$ for some $i$ with $\interp{Cf_i} \subseteq L$.
    Since $f_i$ is fixpoint-free, so is $Cf_i$, and \Cref{cor:fixpointfree:prove} yields $\derive N H Cf_i \leq f$.
    We also have $N \vdash \underline a \leq Ce_i$ by \Cref{lem:expressions can be compared using terms}, thus we get
    \begin{align*}
      \derive N H \underline a \leq Ce_i \leq Cf_i \leq f\enspace.
    \end{align*}
    This shows $a \in L$, hence that $HL \subseteq L$.
\end{proofE}

\section{Reductions}
\label{sec:reductions}
\pratendSetLocal{category=reductions}

We have seen that the naive axiomatisation (\Cref{fig:rules,fig:naive}) is sound with respect to the interpretation $H^*\interp{\_}$ but not always complete.
Nevertheless, there are complete axiomatisations in the literature. That is, for specific signatures $\Sigma$, sets of equations $E$, and sets of hypotheses $H$, there are axiomatisations $Q$, such that for all expressions $e,f$, we have
\begin{equation}
  \label{eq:completefor}
  \tag{$\dagger$}
  \derive{Q}{H} e \leq f \quad\implied\quad \cinterp{e}_E \subseteq \cinterp{f}_E\enspace.
\end{equation}
Following \cite{pous2024tools}, we introduce \emph{reductions} in this section to derive new completeness results from existing ones---typically, to derive that $Q$ is complete for a given set $H$ of hypotheses from the fact that $Q$ is complete for the empty set of hypotheses.

In full generality, reductions are best presented at the fairly abstract level of \emph{representations}~\cite{brunet26:representations}, which we show how to instantiate in the context of the present paper.

\begin{definition}
  A \emph{representation} is tuple $(\F,M,[\_],\equiv)$ where $\F$ and $M$ are two sets, $\equiv$ is an equivalence relation on $\F$, and $[\_]$ is a function from $\F$ to $M$. Such a representation is \emph{complete} if $[e]=[f]$ implies $e\equiv f$ for all $e,f\in \F$.
\end{definition}
In our case, representations will always consist of a fragment $\F$, the set of $H$-closed languages, the interpretation $\cinterp{\_}$, and the equivalence induced by an axiomatisation $Q$ and the hypotheses $H$: $e\equiv f$ if $\derive Q H e \leq f$ and $\derive Q H f \leq e$.
% We assume all such $Q$ are sound.
% : $\derive Q H e \leq f$ implies $\cinterp{e} \subseteq \cinterp{f}$.
For these representations we write $(Q,H)$; they are complete precisely when $Q$ is complete for $H$ in the sense of \eqref{eq:completefor}.

\vspace{\baselineskip}\noindent
\begin{minipage}{.8\linewidth}
  \begin{definition}
    \label{def:reduction}
  A representation $(\F,M,[\_],\equiv)$ reduces to a representation $(\F',M',[\_]',\equiv')$ if there exists three functions
    $(r,i,\transformation)$ typed as on the right
    % $r \colon \F \to \F'$,
    % $i \colon \F' \to \F$,
    % $\transformation \colon M \to M'$,
    such that for all expressions $e \in \F$ and $e',f' \in \F'$,
  \end{definition}
\end{minipage}\hfill
\begin{minipage}{.16\linewidth}
  \vspace{-.5\baselineskip}
  \begin{tikzcd}[cramped]
      {\F} & M \\
      {\F'} & {M'}
      \arrow["{[\_]}", from=1-1, to=1-2]
      \arrow["r", shift left=2, from=1-1, to=2-1]
      \arrow["\transformation", from=1-2, to=2-2]
      \arrow["i", shift left=2, from=2-1, to=1-1]
      \arrow["{[\_]'}", from=2-1, to=2-2]
    \end{tikzcd}
\end{minipage}
\begin{align*}
  \text{ 1) } e' \equiv' f' \,\text{ implies }\, i(e') \equiv i(f');\qquad
  \text{ 2) } i(r(e)) \equiv e;\qquad
  \text{ 3) } \transformation([e]) = [r(e)]'.
\end{align*}
As expected, reductions can be composed, and we have:
\begin{textAtEnd}%
  Before proving the results from \Cref{sec:reductions}, we provide a few more observations about representations and reductions.

  A representation $(\F,M,[\_],\equiv)$ is:
  \begin{itemize}    
  \item \emph{sound} if $e\equiv f$ implies $[e]=[f]$ for all $e,f\in\F$;
  \item \emph{complete} if $[e]=[f]$ implies $e\equiv f$ for all $e,f\in\F$.
  \end{itemize}

  \begin{fact}    
    Let $Q$ and $H$ be an axiomatisation and a set of hypotheses over a fragment $\F$.
    \begin{itemize}
    \item $(Q,H)$ is sound iff $\derive Q H e\leq f$ implies $\cinterp e\subseteq \cinterp f$  for all $e,f\in\F$;
    \item $(Q,H)$ is complete iff $\cinterp e\subseteq \cinterp f$ implies $\derive Q H e\leq f$ implies for all $e,f\in\F$.
    \end{itemize}
  \end{fact}
 
  \begin{theorem}
    Let $Q$ be an axiomatisation over a fragment $\F$. Then $Q$ is sound for continuous models of $E$ (i.e., $\CL \models Q$) if and only if $(Q,H)$ is a sound representation for all $H$. 
  \end{theorem}
  \begin{proof}
    The left-to-right implication is an immediate consequence of \Cref{prop:rules:sound}.
    For the right-to-left implication, pick a sentence $H\rightarrow e\leq f$ from $Q$. We immediately have $\derive Q H e \leq f$, hence $\cinterp e\subseteq \cinterp f$ since $(Q,H)$ is a sound representation by assumption, and thus $\CL\models H\rightarrow e\leq f$ by \Cref{thm:sound:complete}. 
  \end{proof}
  Call \emph{reduction} a tuple $(r,i,\transformation)$ as in \Cref{def:reduction}.  
  \begin{fact}
    Let $R,R',R''$ be representations.
    If $(r,i,\transformation)$ is a reduction from $R$ to $R'$ and
    $(r',i',\transformation')$ is a reduction from $R$ to $R''$, then
    $(r'\circ r,i\circ i',\transformation'\circ \transformation)$ is a reduction from $R$ to $R''$.
  \end{fact}
\end{textAtEnd}%
\begin{theoremE}
  \prooflink
  \label{thm:reductions transfer completeness}
  If a representation reduces to a complete one, then the former is complete.
\end{theoremE}
\begin{proofE}
  Name the components of the representations and the reduction as in \Cref{def:reduction}.
  For all expressions $e,f\in \F$, we have:
  \begin{align*}
    & [e] = [f]\\
    \implies\quad& \transformation([e]) = \transformation([f]) \tag{functions preserve equality}\\
    \implies\quad& [r(e)]' = [r(f)]' \tag{third item of \Cref{def:reduction}}\\
    \implies\quad&r(e) \equiv' r(f) \tag{completeness of target representation}\\
    \implies\quad&i(r(e)) \equiv i(r(f)) \tag{first item of \Cref{def:reduction}}\\
    \implies\quad&e \equiv f \tag*{(second item of \Cref{def:reduction})~\qedhere}
  \end{align*}
\end{proofE}
Most of the tools provided in \cite{pous2024tools} for creating reductions and combining existing ones can be adapted to our general framework, which we do in Appendix~\ref{app:reductions}.
We highlight one of them here, which we need for one of our examples: reductions generated by substitutions, which generalise the \emph{homomorphic reductions} from~\cite[Proposition~3.2]{pous2024tools}.
\begin{textAtEnd}%
  \begin{lemma}%
    \label{lem:red:clos:leq}
    Let $Q,Q'$ be axiomatisations and let $H,H'$ be sets of hypotheses, all over a common fragment.
    If $Q$ is sound for $\CL$ and for all expressions $e,f\in\F$, $\derive{Q'}{H'} e \leq f$ implies $\derive{Q}{H} e \leq f$, then every $H$-closed language is also $H'$-closed, and $H'^*\leq H^*$.
  \end{lemma}%
  \begin{proof}%
    Let $L$ be an $H$-closed language and suppose that $C\interp f\subseteq L$ for some hypothesis $e\leq f$ from $H'$; we have to show that $C\interp e \subseteq L$. First observe that $\derive{Q'}{H'} e\leq f$, hence $\derive{Q}{H} e\leq f$ by assumption, so that $\interp e\subseteq \cinterp f$ since $Q$ is assumed to be sound (using \Cref{prop:rules:sound} and \Cref{thm:sound:complete}).
    We then deduce:
    \begin{align*}    
      C\interp e
      \tag{above observation}
      &\subseteq C\cinterp f\\
      \tag{\Cref{lem:closurefullycontextual}}    
      &\subseteq H^*C\interp f\\
      \tag{by assumption}    
      &\subseteq H^*L\\
      \tag{$L$ is $H$-closed}    
      &= L\enspace.
    \end{align*}
    Thus every $H$-closed language is also $H'$-closed.

    Then, for any language $L$, $H^*L$ is $H$-closed and hence $H'$ closed by the previous point. Since $H^*L$ contains $L$, we deduce that $H'^*L\subseteq H^*L$. We have thus proved $H'^*\leq H^*$.
  \end{proof}%
\end{textAtEnd}%
\begin{propositionE}\label{thm:homomorphic reduction}
  \prooflink
  Let $(Q,H)$ and $(Q',H')$ be representations as above over $\F$, such that $Q$ is sound for $\CL$ and derives all axioms in $\Naivegen$.
  If $\theta \colon X \pto \F$ is a substitution such that:
  \begin{enumerate}
    \item for all expressions $e,f\in\F$, $\derive{Q'}{H'} e \leq f$ implies $\derive{Q}{H} e \leq f$,
    \item for all variables $x$ in the domain of $\theta$, $\derive Q H x = \theta(x)$,
    \item for all variables $x$ in the domain of $\theta$, $x \in \interp{\theta(x)}$,
    \item and for all $(e \leq f) \in H$, $\interp{e\theta} \subseteq H'^*\interp{f\theta}$,
  \end{enumerate}
  then $(Q,H)$ reduces to $(Q',H')$.
\end{propositionE}
\begin{proof}[Proof Sketch]
  We take $(\_)\theta$ for $r$, the identity for $i$ and an inclusion for $\transformation$.
\end{proof}
\begin{proofE}
  We take the substitution application $(\_)\theta$ for $r$, the identity map for $i$, and an inclusion map for $\transformation$. For the latter map, we first have to check that every $H$-closed language is also $H'$-closed, which is the case by \Cref{lem:red:clos:leq} using assumption \textbf{1.}
  
  The first requirement of a reduction (\Cref{def:reduction}) is just assumption \textbf{1.}
  The second one ($\derive Q H \theta(e)=e$) follows from assumption \textbf{2.} via \Cref{lem:extend:id}.
  It remains to show the third one: $H^*\interp{e} = H'^*\interp{e\theta}$.

  For the right-to-left inclusion, we have $\cinterp[H']{e\theta}\subseteq \cinterp{e\theta} = \cinterp e$, using \Cref{lem:red:clos:leq} first, and then the soundness of $Q$ applied to the above derivation $\derive Q H \theta(e)=e$.

  The left-to-right inclusion is more involved. Let $\dot\theta\eqdef\lext\sigma$, with $\sigma$ defined as in \Cref{cor:interp:subst}: $\sigma(x)=\interp{\theta(x)}$ if $\theta(x)$ is defined, and $\sigma(x)=\set x$ otherwise.
  This function $\dot\theta$ is continuous;
  % , and for all expressions $e$, we have $\interp{e\theta}=\dot\theta\interp e$.
  %
  also observe that for all languages $L$, $L \subseteq \dot{\theta}(L)$ by assumption \textbf{3.} We have 
  \begin{align*}
    \cinterp e 
    \tag{last observation}
    &\subseteq \dot\theta \cinterp e \\
    \tag{$\dagger$}
    &\subseteq H'^*\dot\theta\interp e\\
    \tag{\Cref{cor:interp:subst}}
    &= \cinterp[H']{e\theta}\enspace,
  \end{align*}
  thus it remains to show $(\dagger)$, or more generally, $\dot\theta H^*\leq  H'^* \dot\theta$.
  We use \Cref{lem:fixed point respects leq 2} to this end, so it suffices to show $\dot\theta H\leq  H'^* \dot\theta$. 
  Let $L$ be a language. Since $\dot{\theta}$ is continuous, and by definition of the function $H$, we have to show $\dot\theta C\!\interp e \subseteq H'^*\dot\theta L$ for all
  contexts $C$ and hypotheses $(e \leq f)$ from $H$ such that $C\!\interp{f} \subseteq L$.
  We derive
  \begin{align*}
    \dot\theta C\!\interp e
    \tag{\Cref{lem:ctxt:std}}
     &=\dot\theta \interp{\underline C e}\\
    \tag{\Cref{cor:interp:subst}}
     &=\interp{(\underline C e)\theta}\\
    \tag{with $C'$ the context associated to $\underline C\,\theta$}
     &=C'\!\interp{e\theta}\\
    \tag{assumption \textbf{4.}}
     &\subseteq C'\cinterp[H']{f\theta}\\
    \tag{\Cref{lem:closurefullycontextual}}    
     &\subseteq H'^*C'\!\interp{f\theta}\\
     &= \cinterp[H']{(\underline C f)\theta}\\
    \tag{\Cref{cor:interp:subst}}
     &= H'^*\dot\theta\interp{\underline C f}\\
    \tag{\Cref{lem:ctxt:std}}
     &= H'^*\dot\theta C\!\interp{f}\\
    \tag*{(assumption on L)~\qedhere}
     &\subseteq H'^*\dot\theta CL
     \enspace.
  \end{align*}
\end{proofE}
\begin{textAtEnd}
We conclude this appendix on reductions by providing tools to combine them.
Those are essentially the same tools as in~\cite{pous2024tools}, which we generalise from the case of Kleene algebra (hence monoids) to that of axiomatisations over fragments over arbitrary signatures $\Sigma$ and equations $E$.

\begin{proposition}
  \label{lem:combining two reductions}
  Let $Q_1,Q_2,Q_1',Q_2'$ be axiomatisations and 
  let $H_1,H_2,H_1',H_2'$ be sets of hypotheses, all over a common fragment.
  If
  \begin{enumerate}
  \item $(Q_1, H_1)$ reduces to $(Q_1', H_1')$ and $(Q_2, H_2)$ reduces to $(Q_2', H_2')$,
  \item $(H_1 \cup H_2)^* \subseteq H_2^* \circ H_1^*$,
  \item $\derive {Q_1'} {H_1'} e \leq f \implies \derive {Q_2'} {H_2'} e \leq f$ for all expressions $e,f$,
  \end{enumerate}
  then $(Q_1 \cup Q_2, H_1 \cup H_2)$ reduces to $(Q_1' \cup Q_2', H_1' \cup H_2')$
\end{proposition}
\begin{proof}
  We adapt the proof of \cite[Lemma 3.14]{pous2024tools}.
  Throughout this proof, we abbreviate $\derive Q H e = f$ as $\equiv_{Q,H}$, for all $Q,H$.
  Moreover, we write $Q = Q_1 \cup Q_2$, $Q' = Q_1' \cup Q_2'$ and $H = H_1 \cup H_2$, $H' = H'_1 \cup H'_2$.

  Let $r_1$ and $r_2$ be the first components of the two reductions from the statement.
  We claim that $r_2 \circ r_1$ yields the expected reduction, for which we need to check that for all $e,f$,
  \begin{enumerate}
  \item $e \equiv_{Q',H'} f \implies e \equiv_{Q,H} f$,
  \item $r_2(r_1(e)) \equiv_{Q,H} e$,
  \item $\cinterp e = \cinterp[H']{r_2(r_1(e))}$.
  \end{enumerate}
  We proceed point by point.
  \begin{enumerate}
  \item By the first assumption, we have $e \equiv_{Q_i', H_i'} f \implies  \equiv_{Q_i, H_i} f$ for $i = 1,2$.
    From this we conclude $e \equiv_{Q', H'} f \implies e \equiv_{Q,H} f$ by~\Cref{lem:derivation strengthening}.
  \item 
    Using~\Cref{lem:derivation strengthening} we get that $e \equiv_{Q_i, H_i} f$ implies $e \equiv_{Q,H} f$ for $i = 1,2$, from which we deduce using the first assumption again that $e \equiv_{Q,H} r_1(e) \equiv_{Q,H} r_2(r_1(e))$.    
  \item 
    First remark that $H_1'^* \subseteq H_2'^* $ by \Cref{lem:red:clos:leq} and our third assumption.
    This in turn implies $(\dagger)$: $H'^* = (H_1' \cup H_2')^* = H_2'^* $.
    Similarly, since $(Q_2,H_2)$ reduces to $(Q'_2,H'_2)$ we have that $H_2'^* \subseteq H^*_2$ which combined with the above implies $(\ddagger)$: $H_2^* \circ H_1'^* = H_2^*$.
    Second, note that $H_2^* \circ H_1^* \subseteq (H_1 \cup H_2)^* \circ (H_1 \cup H_2)^* = (H_1 \cup H_2)^*$, hence by the second assumption $H^* = (H_1 \cup H_2)^* = H_2^* \circ H_1^*$.
    Using this we deduce
    \begin{align*}
      H^* \interp{e}
      \tag{by assumption}
      &= H_2^*\bigl(H_1^* \interp{e}\bigr)\\
      \tag{first reduction}
      &= H_2^*\bigl(H_1'^* \interp{r_1(e)}\bigr)\\
      \tag{by $(\ddagger)$}
      &= H_2^* \interp{r_1(e)}\\
      \tag{second reduction}
      &= H_2'^* \interp{r_2(r_1(e))}\\
      \tag*{(by $(\dagger)$)~\qedhere}
      &= H'^* \interp{r_2(r_1(e))}
    \end{align*}
  \end{enumerate}
\end{proof}
Usually, the second condition of \Cref{lem:combining two reductions} is the harder one to verify, which is eased by the following four lemmas.
\begin{lemma}[\protect{\cite[Proposition~A.4]{pous2024tools}}]
\label{lem:closurecomp}
Let $H_1, H_2$ be two sets of hypotheses.
We have $(H_1 \cup H_2)^* = H_2^* \circ H_1^*$ if and only if
$H_1^* \circ H_2^* \subseteq H_2^* \circ H_1^*$.
\end{lemma}

\begin{lemma}[\protect{\cite[Proposition A.6]{pous2024tools}}]
Let $H_1, H_2$ be two sets of hypotheses.
If $H_1 \circ H_2^* \subseteq H_2^* \circ H_1^*$ then
$H_1^* \circ H_2^* \subseteq H_2^* \circ H_1^*$.
\end{lemma}

\begin{lemma}[\protect{\cite[Proposition A.12]{pous2024tools}}]
\label{lem:comm:madeeasy}
Let $H_1, H_2$ be two sets of hypotheses.
If $H_1$ is continuous and either
\begin{itemize}
\item $H_1 \circ H_2 \subseteq H_2^{*} \circ H_1$, or
\item $H_1 \circ H_2 \subseteq H_2 \circ H_1^{*}$,
\end{itemize}
then $H_1^* \circ H_2^* \subseteq H_2^* \circ H_1^*$.
\end{lemma}

\begin{lemma}
\label{lem:closcontinuous}
If $H$ is a set of hypotheses whose right-hand sides all denote singleton languages, then the functions $H$ and $H^*$ are continuous.
\end{lemma}
\begin{proof}
  Same argument as for \cite[Lemma 3.1]{pous2024tools}.
\end{proof}
It is also be possible to combine more than two reductions if they share a common target.
\begin{proposition}
  \label{lem:combining more than two reductions}
  Let $Q_0, \dots, Q_n, Q'$ be axiomatisations and let $H_0, \dots, H_n, H'$ be sets of hypotheses, all over a common fragment.
  Let $H = \bigcup_{i \leq n} H_i$, $Q = \bigcup_{i \leq n} Q_i$ and assume $(Q_i, H_i)$ reduces to $(Q',H')$ for all $i$.
  If
  \prooflink
  \begin{itemize}
  \item $H^* \subseteq H_n^* \dots H_0^*$,
  \item $H'^* \subseteq H^*$,
  \end{itemize}
  then $(Q,H)$ reduces to $(Q',H')$.
\end{proposition}
\begin{proof}[Proof sketch]%
  Like above for \Cref{lem:combining two reductions}, we can easily adapt the proof of \cite[Proposition~5.1]{pous2024tools}, using 
  \Cref{lem:red:clos:leq} and \Cref{lem:derivation strengthening} at appropriate places.
\end{proof}%
Like above, the following lemma can be useful to verify the first item of \Cref{lem:combining more than two reductions}.
\begin{lemma}[\protect{\cite[Proposition~A.5]{pous2024tools}}]
  Let $H_0,\dots,H_n$ be sets of hypotheses and write $H_{<j}$ for $\bigcup_{i<j} H_i$.
  If for all $j \leq n$ we have $H_{<j}^* H_j^* \subseteq H_j^* H_{<j}^*$,
  then $H_{<n+1}^* = H_n^* \dots H_0^*$.
\end{lemma}
\end{textAtEnd}

\section{Examples}
\label{sec:examples}
\pratendSetLocal{category=examples}

We finally show how various structures and results from the literature fit our framework, and we prove new completeness results for specific forms of hypotheses.

\subsection{Kleene algebra: monoids}

Our starting point was Kleene algebra with hypotheses~\cite{dkpp:fossacs19:kah,pous2024tools}, which we capture by choosing the signature and equations of monoids for $\Sigma$ and $E$: in that case atoms are just words, and languages are as usual. When it comes to axiomatisations, we chose the fragment where all fixpoints must have the shape $\mu x.1+e\cdot x$, and we may pick any of the various axiomatisations proposed in the literature~\cite{conway1971regular,boffa1990remarque,Krob91a,Kozen91,Boffa95,KozenS12,ddp:lpar18:lefthanded,Kappe23}, which are all sound, and complete for the empty set of hypotheses.

For hypotheses, our notion of language closure coincides with that from~\cite{dkpp:fossacs19:kah,pous2024tools}, and the reductions proposed in the literature fit our definition and make it possible to reach KA with tests~\cite{kozen1997tests}, with converse~\cite{EB95,bloom1995notes,brunet2014converse}, with top~\cite{pous2022top,PousW24}, with observations~\cite{pous2024tools}, or even NetKAT~\cite{AndersonFGJKSW14,pous2024tools} and synchronous KA~\cite{Prisacariu10,WagemakerBKR019}.
(Note that some of these reductions depend on the specific choice of axiomatisation for KA without hypotheses; e.g., the two items of \cite[Proposition~3.8(iv)]{pous2024tools} respectively require left-handed KA and right-handed KA.)

More concretely, we recall here a simple example~\cite[Proposition~3.8(v)]{pous2024tools}, which we will transfer to a new setting in the following section. This example shows that we can eliminate ``transitivity'' hypotheses of the shape $a\cdot a\leq a$, in KA.
\begin{lemma}
  \label{ex:KA:transitive}
   $(\text{KA},\set{a\cdot a\leq a})$ reduces to $(\text{KA},\emptyset)$.
\end{lemma}
\begin{proof}
  The substitution $a\mapsto a\cdot a^*$ satisfies the conditions of \Cref{thm:homomorphic reduction}:
  the first one vanishes as $Q=Q'=\text{KA}$ and $H'=\emptyset$; we can easily derive $a=a\cdot a^*$ from KA and the hypothesis $a\cdot a\leq a$; we have $a\in \interp{a\cdot a^*}$; and we have
  $\interp{(a\cdot a^*)\cdot(a\cdot a^*)}\subseteq \interp{a\cdot a^*}$.
\end{proof}

\subsection{Bi-Kleene algebra: bimonoids}

We also cover bi-KA with hypotheses~\cite{KappeB0WZ20,kappethesis,wagemakerthesis}:
we chose the signature and equations of bimonoids, atoms become series-parallel pomsets, and our notion of language closure coincides with that from~\cite{KappeB0WZ20,kappethesis,wagemakerthesis}.

Regarding syntax, there is a choice: while Struth and Laurence~\cite{struth2014bikleene} prove completeness for bi-KA with two fixpoint operators (the sequential one from KA, and its parallel counterpart: $e^{\parallel}\eqdef\mu x.1+e\parallel x$), Kappe et al.~\cite{KappeB0WZ20,kappethesis,wagemakerthesis} develop the theory of hypotheses and reductions only with the sequential one. This restriction is necessary for the reduction they exhibit to deal with the exchange law and prove completeness of concurrent KA~\cite{hoare2011concurrent,kappe2018concurrent}; it is unclear how to extend this reduction in the presence of unbounded parallelism.

We cover both choices with the present framework, by selecting the appropriate fragment.
If we choose bi-KA with the two fixpoint operators~\cite{struth2014bikleene}, for which hypotheses were not studied in the literature before, we can redo \Cref{ex:KA:transitive} with both monoid compositions, in a straightforward manner. 
\begin{lemma}
  \label{ex:biKA:transitive}
  $(\text{biKA},\set{a\cdot a\leq a})$ and $(\text{biKA},\set{a\parallel a\leq a})$ both reduce to $(\text{biKA},\emptyset)$.
\end{lemma}
\begin{proof}
  We use \Cref{thm:homomorphic reduction} with
  the substitutions $a\mapsto a\cdot a^*$ and $a\mapsto a\parallel a^{\parallel}$\enspace.
\end{proof}

\begin{textAtEnd}
  Throughout this section, we abbreviate $\derive Q H e\leq f$ to $e\leqq_{Q,H} f$ and 
  $\derive Q H e\leq f$ and $\derive Q H f\leq e$ together to $e\equiv_{Q,H}f$. Whenever $Q$ or $H$ are clear from context we omit (part of) the subscript.
\end{textAtEnd}
Modularity tools from~\cite{pous2024tools}, which we have adapted in Appendix~\ref{app:reductions}, make it possible to deduce that $(\text{biKA},\set{a\cdot a\leq a, b\parallel b\leq b})$ combinedly reduces to $(\text{biKA},\emptyset)$, when $a$ and $b$ are distinct variables (see details in Appendix~\ref{bikamodular}).
In contrast, there cannot be a reduction from $(\text{biKA},\set{a\cdot a\leq a, a\parallel a\leq a})$ to $(\text{biKA},\emptyset)$: this would require an expression for all non-empty series-parallel pomsets over $\set a$, which is not available in the considered fragment.
\begin{textAtEnd}
  \stepcounter{subsection}
  \subsection{Bi-Kleene algebra}\label{bikamodular}
  We show here how to use \Cref{lem:combining two reductions} in order to combine the two reductions established in~\Cref{ex:biKA:transitive}.
  \begin{lemma}
  $(\text{biKA},\set{a\cdot a\leq a, b\parallel b\leq b})$ reduces to $(\text{biKA},\emptyset)$.
  \end{lemma}
  \begin{proof}
    We use~\Cref{lem:combining two reductions}. The first condition follows from~\Cref{ex:biKA:transitive} and the last condition is trivially satisfied. The second condition can be proved using~\Cref{lem:closurecomp}, where it remains to show that for $H_1=\set{a\cdot a\leq a}$ and $H_2=\set{b\parallel b\leq b}$ we have $H_1^*\circ H_2^*\subseteq H_2^*\circ H_1^*$.
    For this we can use \Cref{lem:comm:madeeasy}:
    we easily get $H_1\circ H_2\subseteq H_2\circ H_1$ since $a\neq b$, and $H_1$ is continuous by \Cref{lem:closcontinuous}.
  \end{proof}
\end{textAtEnd}

\subsection{Regular tree languages}

When $E$ is empty, atoms are just terms, or ordered trees with nodes labelled in $\Sigma$ and of appropriate arities. In that case, there is a Kleene theorem~\cite{TATAhubert} ensuring that our expressions and finite tree automata capture the same class of tree languages: the regular ones (note that we need all fixpoint expressions in order to obtain this result, unlike with languages of words).
Moreover, as mentioned in \Cref{sec:axiomatisation}, Ésik~\cite{esik2010tree} has proved that the naive axiomatisation $N(\emptyset,\expr X)$ is complete for the empty set of hypotheses in that case (\Cref{thm:completeness for no equations}).

Hypotheses were not studied in the literature in that setting, and we show below how to eliminate the counterpart to \emph{Hoare hypotheses}, from~\cite{hardin2002elimination} in the context of KAT.

\begin{textAtEnd}
  \subsection{Regular tree languages}
\end{textAtEnd}

\begin{lemmaE}
  \label{ex:e_0tree}
  \prooflink
  Let $e_0$ be an expression. If both the signature $\Sigma$ and the set $X$ of variables are finite, then $(N(\emptyset,\expr X),\set{e_0\leq 0})$ reduces to $(N(\emptyset,\expr X),\emptyset)$.
\end{lemmaE}
\begin{proof}[Proof Sketch]
  Using the aforementioned Kleene theorem, let $\top$ be an expression for the language of all trees over $\Sigma$ and $X$.
  Define the following expression:
  \begin{align*}
    e_0^\top := \mu x.e_0+\sum_{s^{(n)} \in \Sigma,~n\geq 1} s(x, \top, \dots, \top) + s(\top, x, \dots, \top) + \dots + s(\top, \top, \dots, x)\enspace.
  \end{align*}
  The language $\interp{e_0^\top}$ consists of all trees containing a tree from $\interp{e_0}$ as a subtree, and the function $e\mapsto e+e_0^\top$ yields the desired reduction.
  We give more details in Appendix~\ref{app:examples}; the key observations are that
   $e_0^\top$ is provably equivalent to $0$ using the naive axioms and the hypothesis $e_0\leq 0$, and that
  writing $H$ for $\set{e_0\leq 0}$, we have $H^*= L\mapsto L \cup \interp{e_0^\top}$.
\end{proof}
\begin{proofE}
  Let $H=\set{e_0\leq 0}$.
  Let $\top$ be an expression for the language of all trees over $\Sigma$ and $X$.
  Pick a recursion variable $x$ and define the following expressions:
  \begin{align*}
    g &\eqdef \sum_{s^{(n)} \in \Sigma,~n\geq 1} s(x, \top, \dots, \top) + s(\top, x, \dots, \top) + \dots + s(\top, \top, \dots, x)\\
    e_0^\top &\eqdef \mu x.e_0+g\enspace.
  \end{align*}
  (Note that $g$ actually is an $x$-expression.)
  Take $r(e) = e + e_0^\top$, the identity function for $i$, and the inclusion $H^*\lang X\subseteq \lang X$ for $\transformation$.
  We have to check that for all expressions $e,f$,
    \begin{itemize}
      \item $e \equiv f \implies e \equiv_{H} f$\:;
      \item $r(e) \equiv_{H} e$\:;
      \item $H^*\interp{e} = \interp{r(e)}$.
    \end{itemize}
    The first point trivially holds by \Cref{lem:derivation strengthening}.
    For the second point, it suffices to show $e_0^\top\leqq_H 0$. As $e_0^\top$ is a least fixpoint, we can conclude $e_0^\top\leqq_H 0$ from $\sub{e_0+g}{x}{0}\leqq_H 0$.
      We derive
      \begin{align*}
        \sub{e_0 + g}{x}{0}
        &\leqq_H 0 + \sub{g}{x}{0}\\
        &\leqq_H \sum_{s^{(n)} \in \Sigma,~n\geq 1} s(0, \top, \dots, \top) + \dots + s(\top, \top \dots, 0)\\
        &\leqq_H \sum_{s^{(n)} \in \Sigma,~n\geq 1} 0 + \dots + 0\\
        &\leqq_H 0\enspace.
      \end{align*}
      % By the fixed point rule we conclude $\derive{N}{H} e_0^\top \leq 0$.
  For the third point, we observe $H^*L = L \cup \bigcup \{ C\interp{e_0} \mid C \text{ a context}\}$ and $\interp{e_0^\top} = \bigcup \{ C\interp{e_0} \mid C \text{ a context}\}$ by construction.
  Hence $H^*\interp{e} = \interp{e} \cup \interp{e_0^\top} = \interp{e + e_0^\top} = \interp{r(e)}$.
\end{proofE}

\subsection{Commutative Kleene algebra: multisets}
\label{sec:examples:CKA}

Commutative Kleene algebras (cKA) are obtained from KA by postulating that $(\cdot)$ is commutative.
Pilling has proved their completeness w.r.t. the Parikh image interpretation of regular expressions~\cite{pilling1970algebra,conway1971regular,brunet2019commutative}.
It is tempting to deal with cKA by using plain KA with the hypotheses $C=\set{xy\leq yx\mid x,y\in X}$. Indeed those hypotheses imply that any two expressions commute, and $C^*$ is the Parikh image function in that case. However, we have $\cinterp[C]{(ab)^*}=\set{w \in \set{a,b}^* \mid |w|_a = |w|_b}$, so that $C^*$ does not preserve regularity. Thus, there cannot be a reduction from $(\text{KA},C)$ to $(\text{KA},\emptyset)$, and we may hardly find reductions to $(\text{KA},C)$.
Instead, we can impose the commutativity condition via the set $E$ of equations, and work directly with commutative monoids, for which atoms are finite multisets. 
\begin{textAtEnd}
  \subsection{Commutative Kleene algebra}
  Recall for commutative Kleene algebra atoms are finite multisets, or commutative words.
  We will write such atom as words where we are free to reorder the letters as we see fit, e.g., $abc=cba=acb$.
\end{textAtEnd}

Below we provide three new reductions; the first one is about Hoare hypotheses.
\begin{lemmaE}
  \label{ex:e<0reductionCKA}
  \prooflink
  Let $e_0$ be an expression. For finite $X$, $(\ComKA,\set{e_0\leq 0})$ reduces to $(\ComKA,\emptyset)$.
\end{lemmaE}
\begin{proof}[Proof sketch]
  Let $\top$ be an expression for the language of all multisets (for which we need $X$ to be finite).
  The function $e\mapsto e+e_0\cdot \top$ yields the desired reduction.
\end{proof}
\begin{proofE}
  Let $H = \{e_0 \leq 0\}$ for some expression $e_0$.
  We define the expression $\top := (a + b + \dots)^*$, where $a,b,\dots$ are the variables in $X$, whose standard language interpretation is the full language of all multisets.
  Take $r(e) = e + e_0\cdot\top$, the identity function for $i$, and the inclusion $H^*\lang X\subseteq \lang X$ for $\transformation$.
  We need to verify that for all $e,f$:
  \begin{itemize}
  \item $e \equiv f \implies e \equiv_{H} f$
  \item $ r(e) \equiv_{H} e$
  \item $H^*\interp{e} = \interp{r(e)}$,
  \end{itemize}
  The first point holds trivially by \Cref{lem:derivation strengthening}.
  For the second point we derive
  \begin{align*}
    r(e)=e + e_0 \cdot \top \equiv_H e + 0 \cdot \top \equiv e + 0 \equiv e\enspace.
  \end{align*}
  For the third point, we note $H^*L = L \cup \{ C \interp{e_0} \mid C \text{ a context} \} = L \cup \interp{e_0 \top}$.
  Hence we have $H^*\interp{e} = \interp{e} \cup \interp{e_0\cdot \top} = \interp{e + e_0 \cdot \top} = \interp{r(e)}$.
  % We conclude $(r,\id,\id)$ is a reduction from $(\ComKA, H)$ to $(\ComKA, \emptyset)$.\\
\end{proofE}

The second one looks similar but has no counterpart in the literature: while it is relatively easy to obtain in the commutative case, we cannot obtain such a reduction in the non-commutative case (e.g., with $e_1=a\cdot b$, regularity is not preserved).
\begin{lemmaE}
  \label{ex:e<1reductionCKA}
  \prooflink
  Let $e_1$ be an expression. $(\ComKA,\set{e_1\leq 1})$ reduces to $(\ComKA,\emptyset)$.
\end{lemmaE}
\begin{proof}[Proof sketch]
  The function $e\mapsto e\cdot e_1^*$ yields the desired reduction.
\end{proof}
\begin{proofE}
  Let $H = \{e_1 \leq 1\}$ for some expression $e_1$.
  Take $r(e) = e\cdot e_1^*$, the identity function for $i$, and the inclusion $H^*\lang X\subseteq \lang X$ for $\transformation$.
    Again, we check that for all $e,f \in \F$:
    \begin{itemize}
      \item $e \equiv f \implies e \equiv_{H} f$ ;
      \item $ r(e) \equiv_{H} e$ ;
      \item $H^*\interp{e} = \interp{r(e)}$.
    \end{itemize}  
    The first point is again trivial.
    For the second point we derive
    \begin{align*}
      e \equiv e \cdot 1 \equiv e \cdot 1^* \equiv_H e \cdot e_1^* = r(e)\enspace.
    \end{align*}
    For the third point, we first verify $\interp{r(e)} \subseteq H^*\interp{e}$.
    By the second point we have $r(e) \leqq_H e$.
    Using that $\ComKA$ is sound, we can use~\Cref{prop:rules:sound} to conclude that $\CL\models H\rightarrow r(e)\leq e$ and \Cref{thm:sound:complete} to obtain $\cinterp{r(e)} \subseteq \cinterp{e}$ and thus $\interp{r(e)} \subseteq H^*\interp{e}$.
    
    For the other direction we show that $\interp{r(e)}$ is $H$-closed:
    \begin{align*}
      H \interp{r(e)} = \{ C \interp{e_1} \mid C \interp{1} 
      \subseteq \interp{e \cdot e_1^*}, C \text{ a context} \}
      \subseteq\interp{e \cdot e_1^*} 
      = \interp{r(e)}\enspace.
    \end{align*}
    The inclusion above is justified by the fact that the application of a multiset context to a multiset is multiset union; in particular $C \interp{1} = \set{C}$.
    Since $\interp{e} \subseteq \interp{r(e)}$, it follows that $H^*\interp{e} \subseteq \interp{r(e)}$. 
\end{proofE}

The third one is a first step towards a counterpart to \cite[Lemma~3.8(i)]{pous2024tools} for the commutative case, which is surprisingly much harder due to the lack of a Kleene theorem.
% (Note that the order of the inequation is reversed w.r.t. \Cref{ex:KA:transitive})

\begin{textAtEnd}
  Now we move to the proof of \Cref{ex:a<aareduction}, for which we need a few preliminary results.
  \begin{fact}\label{fact:cka}
    For all expressions $e,f \in \F$, we have $e^*f^* \equiv_{\ComKA} (e+f)^*$.
  \end{fact}

\begin{lemma}
  \label{lem:sumandprodunderstar}
  For all expressions $e_0, \dots, e_n$ we have $e_0^*\dots e_n^* \equiv_{\ComKA} (e_0 + \dots + e_n)^*$.
\end{lemma}
\begin{proof}
\Cref{fact:cka} generalises to the required result.
\end{proof}

In the lemmas that follow, we let $H=\set{a\leq a\cdot a}$;
also recall we have defined the notation $a^{\leq k} = 1 + a + \dots a^k$.
\begin{lemma}
  \label{lem:akbelow}
  For all $k \geq 0$, $aa^{\leq k} \equiv_{\ComKA, H} a^{k + 1}$.
\end{lemma}
\begin{proof}
  Observe that $a a^{\leq k} \equiv a+aa+\dots + a^{k+1}$. Via repeated applications of $H$, we know that $a^j\leqq a^k$ for all $k\geq 1$ and $1\leq j\leq k$, from which we can conclude $aa^{\leq k} \leqq a^{k + 1}$. It is trivial to prove $a^{k + 1}\leqq aa^{\leq k}$.\qedhere
\end{proof}

\begin{lemma}
  \label{lem:akstarbelow}
  For all $k\geq 0$ and for all terms $u$, we have $a(ua^k)^* \equiv_{\ComKA,H} a(ua^{\leq k})^*$.
\end{lemma}
\begin{proof}
  First, we observe $a(ua^k)^* \leqq a(ua^{\leq k})^*$ since $a^k \leqq a^{\leq k}$.
  
  Next, we prove $a(ua^{\leq k})^* \leqq a(ua^k)^*$ using the following implicational axiom from commutative KA:
  $$
    x + y \cdot z \leqq y \implies x \cdot z^* \leqq y.
    $$
  To conclude, we need to derive
    \begin{align*}
      a + a(ua^k)^*(ua^{\leq k})
    &\leqq   a + uaa^{\leq k}(ua^k)^*  \tag{$ \ComKA \vdash e\cdot f = f\cdot e \quad \forall e,f$}\\
    % &\equiv  a + uaa^{\leq k}(ua^k)^*\\
    &\leqq  a + ua^{k+1}(ua^k)^* \tag{\Cref{lem:akbelow}}\\
    &\leqq  a(1 + ua^{k}(ua^k)^*)\\
    &\leqq  a(ua^k)^* \tag{$ \ComKA \vdash e^* = 1+e\cdot e^* \quad \forall e$}
  \end{align*}
\end{proof}

\end{textAtEnd}

\begin{lemmaE}
  \label{ex:a<aareduction}
  \prooflink
   $(\ComKA,\set{a\leq a\cdot a})$ reduces to $(\ComKA,\emptyset)$.  
\end{lemmaE}
\begin{proof}[Proof sketch]
  Let $H \eqdef \{a \leq a\cdot a\}$. For all languages $L$, we have
  $$
  H^*L = L \cup \{ a^ku \mid a^{k'}u \in L, 1 \leq k < k' \}.
  $$
  In cKA, all expressions are provably equivalent to an expression of the form $e_0 + \dots + e_n$, where $e_i$ is of the form $w_0 w_1^* \dots w_n^*$ and all $w_i$ are distinct finite multisets~\cite{pilling1970algebra, conway1971regular}. 
 For $H^*$, $H^*(L \cup L') = H^*(L) \cup H^*(L')$ (cf.\ \cite[Lemma 3.1]{pous2024tools} for a similar argument); hence we construct an $e_i'$ for each $e_i$ and define $r(e)=e_0'+\dots + e_n'$.

  For $e_i$ of the form $w_0 w_1^* \dots w_n^*$, we treat the cases where $a \in w_0$, and $a \not\in w_0$ separately. The latter is in the appendix. 
  If $a \in w_0$, we use commutativity to obtain the equivalent expression
  $
   (u_0a^{k_0})( u_1 a^{k_1} )^* \cdots (u_m a^{k_{m}} )^*v_0^* \cdots v_{m'}^*
  $
  with each $k_i\geq 1$.
  Then we define
  \begin{align*}
    e_i' \eqdef (u_0 aa^{\leq k_0 - 1})( u_1 a^{\leq k_1} )^* \cdots (u_m a^{\leq k_{m}} )^*v_0^* \cdots v_{m'}^*,
  \end{align*}
  where $a^{\leq k} \eqdef 1 + a + \dots + a^k$.

  We show that $r$ yields a reduction in the appendix; here we provide intuition through an example.
  Let $e=aa(aa)^*$, whose language is $\set{ a^k \mid k\text{ even and } k\geq 2 }$ with $H$-closure
  $
    \set{ a^k\mid k\geq 1 }
  $. Using the naive axioms we derive that $e'\eqdef(a(1+a))(1+a+aa)^* =(a+aa)(1+a+aa)^*$, and then it is easy to see that the language interpretation of $e'$ coincides with the $H$-closure. We finally derive from the naive axioms and the hypothesis that $e=e'$:
  \begin{align*}
    aa(aa)^* \leq (a+aa)(1+a+aa)^* \leq (aa+aa)(1+aa+aa)^* =aa (1+aa)^*=aa(aa)^* \tag*\qedhere
  \end{align*}
\end{proof}
\begin{proofE}
  In cKA, all expressions are provably equivalent to an expression of the form $e_0 + \dots + e_n$, where $e_i$ is of the form $w_0 w_1^* \dots w_n^*$ and all $w_i$ are distinct finite multisets~\cite{pilling1970algebra, conway1971regular}. 

  For $e_i$ of the form $w_0 w_1^* \dots w_n^*$, if $a \in w_0$, we use commutativity to obtain the equivalent expression
  $
   (u_0a^{k_0})( u_1 a^{k_1} )^* \cdots (u_m a^{k_{m}} )^*v_0^* \cdots v_{m'}^*
  $
  with each $k_i\geq 1$.
  Then we define
  \begin{align*}
    e_i' \eqdef (u_0 aa^{\leq k_0 - 1})( u_1 a^{\leq k_1} )^* \cdots (u_m a^{\leq k_{m}} )^*v_0^* \cdots v_{m'}^*,
  \end{align*}

  On the other hand, if $a\notin w_0$, then $e_i$ is equivalent to $w_0( u_1 a^{k_1} )^* \dots (u_m a^{k_{m}} )^*v_0^* \dots v_{m'}^*$ and we define
    $$
      e_i'\eqdef w_0 \big(1 + (u_1 aa^{\leq k_1 - 1} + \dots + u_m aa^{\leq k_m - 1})( u_1 a^{\leq k_1} )^* \cdots (u_m a^{\leq k_{m}} )^*\big)v_0^* \cdots v_{m'}^*.
    $$

Take $r(e)= e_0' + \dots + e_n'$, $i$ the identity function and $\transformation$ an inclusion. We must check:
    \begin{itemize}
      \item $e \equiv f \implies e \equiv_{H} f$\:;
      \item $ r(e) \equiv_{H} e$\:;
      \item $H^*\interp{e} = \interp{r(e)}$.
    \end{itemize} 
The first point is satisfied trivially by \Cref{lem:derivation strengthening}. For the second point, it suffices to show that $e_i \equiv_H e_i'$. We first consider the case where $e_i$ is provably equivalent to $w_0 w_1^* \dots w_n^*$ with $a\in w_0$. What remains to show is 
$$
(u_0a^{k_0})( u_1 a^{k_1} )^* \cdots (u_m a^{k_{m}} )^*\equiv_H(u_0 aa^{\leq k_0 - 1})( u_1 a^{\leq k_1} )^* \cdots (u_m a^{\leq k_{m}} )^*
$$

We just show the first step of the derivation:
\begin{align*}
 (u_0a^{k_0})( u_1 a^{k_1} )^* 
      &\equiv_H (u_0aa^{\leq k_0 - 1})( u_1 a^{k_1} )^* \tag{\Cref{lem:akbelow}}\\
      &\equiv_H (u_0a^{\leq k_0 - 1})a( u_1 a^{k_1} )^* \tag{$ \ComKA \vdash e\cdot f = f\cdot e \quad \forall e,f$}\\
      &\equiv_H (u_0a^{\leq k_0 - 1})a( u_1 a^{\leq k_1} )^* \tag{\Cref{lem:akstarbelow}}
\end{align*}

We can then repeatedly move the single $a$ to the right using commutativity and apply~\Cref{lem:akstarbelow}, and eventually move that single $a$ back to $u_0a^{\leq k_0 - 1}$ to obtain the result.

For $e_i$ provably equivalent to $w_0 w_1^* \dots w_n^*$ with $a\notin w_0$, what remains to show is
$$
( u_1 a^{k_1} )^* \cdots (u_m a^{k_{m}} )^*\equiv_H 1 + (u_1 aa^{\leq k_1 - 1} + \dots + u_m aa^{\leq k_m - 1})( u_1 a^{\leq k_1} )^* \dots (u_m a^{\leq k_{m}} )^* 
$$
    We show the derivation for $m=2$, but the derivation is the same for any $m$:
    \begin{align*}
      ( u_1 a^{k_1} )^* \cdot (u_2 a^{k_{2}} )^* 
      &\equiv_H ( u_1 a^{k_1}  + u_2 a^{k_{2}} )^* \tag{\Cref{lem:sumandprodunderstar}}\\
      &\equiv_H 1 + (u_1 a^{k_1}   + u_2 a^{k_{2}}) ( u_1 a^{k_1}  + u_2 a^{k_{2}} )^* \tag{$ \ComKA \vdash e^* = 1+e\cdot e^* \quad \forall e$}\\
      &\equiv_H 1 + (u_1 a^{k_1} + u_2 a^{k_{2}}) ( u_1 a^{k_1} )^* \cdot (u_2 a^{k_{2}} )^* \tag{\Cref{lem:sumandprodunderstar}}\\
      &\equiv_H 1 + a(u_1 a^{k_1 - 1}  + u_2 a^{k_{2} - 1}) ( u_1 a^{k_1} )^* \cdot (u_2 a^{k_{2}} )^* \\
      &\equiv_H 1 + (u_1 a^{k_1 - 1} + u_2 a^{k_{2} - 1}) a( u_1 a^{k_1} )^* \cdot (u_2 a^{k_{2}} )^* \\    
      &\equiv_H 1 + (u_1 a^{k_1 - 1}   + u_2 a^{k_{2} - 1}) a( u_1 a^{\leq k_1} )^* \cdot (u_2 a^{k_{2}} )^* \tag{\Cref{lem:akstarbelow}}\\
      % &\:\:\:\vdots\\
      % &\equiv_H 1 + (u_1 a^{k_1 - 1} + \dots  + u_m a^{k_{m} - 1}) ( u_1 a^{\leq k_1} )^* \dots a(u_m a^{k_{m}} )^* \\
      % &\equiv_H 1 + (u_1 a^{k_1 - 1}   + u_2 a^{k_{2} - 1}) ( u_1 a^{\leq k_1} )^* \cdot a(u_2 a^{\leq k_{2}} )^* \\
      &\equiv_H 1 + (u_1 a^{k_1 - 1}   + u_2 a^{k_{2} - 1}) ( u_1 a^{\leq k_1} )^* \cdot a(u_2 a^{\leq k_{2}} )^* \tag{\Cref{lem:akstarbelow}}\\
      % &\equiv_H 1 + a(u_1 a^{k_1 - 1}   + u_2 a^{k_{2} - 1}) ( u_1 a^{\leq k_1} )^* \cdot (u_2 a^{\leq k_{2}} )^* \\
      &\equiv_H 1 + (u_1 a^{k_1}  + u_2 a^{k_{2}}) ( u_1 a^{\leq k_1} )^* \cdot (u_2 a^{\leq k_{2}} )^* \\
      &\equiv_H 1 + (u_1 aa^{\leq k_1 - 1} + u_2 aa^{\leq k_{2} - 1}) ( u_1 a^{\leq k_1} )^* \cdot (u_2 a^{\leq k_{2}} )^* \tag{\Cref{lem:akbelow}}
    \end{align*}

  For the third point, we first show that $\interp{e_i'} = H^*\interp{e_i}$ for all $i$.
  Observe that for all languages $L$, we have
  $$
  H^*L = L \cup \{ a^ku \mid a^{k'}u \in L, 1 \leq k < k' \}.
  $$
    If $w_0$ contains $a$, then $e_i$ is provably equivalent to $
   (u_0a^{k_0})( u_1 a^{k_1} )^* \cdots (u_m a^{k_{m}} )^*v_0^* \cdots v_{m'}^*
  $, so we compute
    $
    H^* \interp{e_i} \equiv \set{ a^k u_0 u_1^{x_1} \dots u_m^{x_m}v_0^{y_0} \dots v_{m'}^{y_{m'}} \mid 1 \leq k \leq k_0 + k_1x_1 + \dots + k_mx_m, x_i, y_i \geq 0}.
    $
    In this case
    $
   e_i' \eqdef (u_0 aa^{\leq k_0 - 1})( u_1 a^{\leq k_1} )^* \dots (u_m a^{\leq k_{m}} )^*v_0^* \dots v_{m'}^*
    $.

    Observe that 
    $$\interp{(u a^{\leq k})^*} = \set{ u^x a^z \mid 0 \leq x, 0 \leq z \leq xk }$$.
    Hence
    $$
    \interp{e_i'} 
    = \set{ (u_0 a^{z_0})( u_1^{x_1} a^{z_1} ) \dots ( u_m^{x_m} a^{z_m} )v_0^{y_0} \dots 
    v_{m'}^{y_{m'}}
    \mid
    1 \leq z_0 \leq k_0, x_i, y_i \geq 0, 0 \leq z_i \leq k_ix_i
    }.
    $$
    Grouping together the $a$'s we obtain
    $$
    \interp{e_i'} 
    = \set{ a^{ z_0 + z_1 + \dots z_m}u_0u_1^{x_1} \dots u_m^{x_m} v_0^{y_0} \dots v_{m'}^{y_{m'}}
    \mid
    1 \leq z_0 \leq k_0, x_i, y_i \geq 0, 0 \leq z_i \leq k_ix_i
    }.
    $$
    Since $1 \leq z_0 \leq k_0$ and $0 \leq z_i \leq k_ix_i$, we have that $1 \leq z_0 + z_1 + \dots z_m \leq  k_0 + k_1x_1 + \dots + k_mx_m$, so $\interp{e_i'} = H^*\interp{e_i}$.
  
    If $w_0$ does not contain $a$ we have $e_i$ provably equivalent to $ w_0( u_1 a^{k_1} )^* \dots (u_m a^{k_{m}} )^*v_0^* \dots v_{m'}^*$ and we compute
    $$
      H^* \interp{e_i} = \set{ w_0a^k u_1^{x_1} \dots u_m^{x_m}v_0^{y_0} \dots v_{m'}^{y_{m'}} \mid 1 \leq k \leq k_1x_1 + \dots + k_mx_m, x_i, y_i \geq 0}.
    $$
    In this case 
    $$
        e_i' \eqdef w_0 \big(1 + (u_1 aa^{\leq k_1 - 1} + \dots + u_m aa^{\leq k_m - 1})( u_1 a^{\leq k_1} )^* \dots (u_m a^{\leq k_{m}} )^*\big)v_0^* \dots v_{m'}^*,
    $$
    and by similar computation as before we have that $\interp{e_i'} = H^*\interp{e_i}$.
    
    For $H^*$, $H^*(L \cup L') = H^*(L) \cup H^*(L')$~(\Cref{lem:closcontinuous}). Thus we can derive that 
    \begin{align*}
      \interp{r(e)} = \interp{e_0'} \cup \dots \cup \interp{e_n'} 
      % \interp{e_0' + \dots + e_n'} 
      % &= \interp{e_0'} \cup \dots \cup \interp{e_n'} \\
      = H^*\interp{e_0} \cup \dots \cup H^*\interp{e_n} 
      = H^*(\interp{e_0} \cup \dots \cup \interp{e_n}) 
      % = H^*\interp{e_0 + \dots + e_n}
      = H^*\interp{e}\enspace.\tag*\qedhere
    \end{align*}  
\end{proofE}

\noindent Let us finally illustrate how we can assemble the previous reductions in a modular way.
Using the tools from Appendix~\ref{app:reductions}, we can strengthen the first two examples above into:
\begin{lemmaE}\label{lem:combine}
  \prooflink
    For all set of hypotheses $H$ and all expressions $e$, the sets $(\ComKA,H \cup \{e \leq 0\})$ and $(\ComKA,H \cup \{e \leq 1\})$ both reduce to $(\ComKA,H)$.
\end{lemmaE}
\begin{proofE}
    This is an application of \Cref{lem:combining two reductions}. 
    We take $H_1 = \{e \leq 0\}$ or $H_1 = \set{e \leq 1}$ where the reduction is provided in \Cref{ex:e<0reductionCKA} and \Cref{ex:e<1reductionCKA} respectively.
    To apply \Cref{lem:combining two reductions}, the only condition left to verify is $(H_1 \cup H)^* \subseteq H^* \circ H_1^*$, which is equivalent to $H_1^* \circ H^* \subseteq H^* \circ H_1^*$ by \Cref{lem:closurecomp}.
    For $H_1 = \set{e \leq 0}$, we note $H_1^*(L) = L \cup \interp{e \cdot \top}$ and compute
    \begin{align*}
      H_1^* (H^* L) =  H^*(L) \cup \interp{e \cdot \top} \subseteq H^*(L) \cup H^*\interp{e \cdot \top} \subseteq H^*(L \cup \interp{e \cdot \top}) = H^* (H_1^* L)
    \end{align*}
    For $H_1 = \set{e \leq 1}$, we note $H_1^*(L) = L \cdot \interp{e^*}$ and compute via~\Cref{lem:closure commutes with distributive symbol} that
    \begin{align*}
      H_1^* (H^* L) =  H^*(L) \cdot \interp{e^*} \subseteq H^*(L \cdot \interp{e^*}) = H^* (H_1^* L)\enspace.\tag*\qedhere
    \end{align*}
\end{proofE}
Since reductions can be composed, we deduce the following corollary:
\begin{corollary}
  For all $e_0,e_1$, $(\ComKA, \set{e_0 \leq 0, e_1 \leq 1, a \leq a\cdot a})$ reduces to $(\ComKA,\emptyset)$.
\end{corollary}

\section{Conclusion}
\label{sec:ccl}
We have developed a framework to study algebras ordered by complete lattices in the presence of hypotheses. This framework captures existing work on Kleene algebra and bi-Kleene algebra with hypotheses, and can be used to study new structures in the presence of hypotheses. In particular, we instantiated our framework to regular tree languages and commutative KA and proved new completeness results for particular hypotheses. These structures had not been studied with additional hypotheses in the literature. 

We conjecture that a completeness result can be obtained for commutative KA and regular tree languages for the hypothesis $a\leq w$, where $a$ is a letter and $w$ a multiset or a tree respectively, for which we can use the developed framework and theory. In addition, we want to instantiate our framework to Kleene algebra with additional continuous symbols to prove a completeness result that captures common preliminary results in the completeness proofs of bi-KA, concurrent KA and identity-free Kleene lattices~\cite{struth2014bikleene,dp:concur18:kl}. Subsequently we can then study these algebras in the presence of hypotheses using the developed theory. 

Finally, we want to investigate the possibility of varying the signature of algebras when providing reductions, where in this paper we have left it fixed.

\bibliography{long,main}

\clearpage
\appendix

\section{Note on variables}
\label{app:variables}

To alleviate notation, we have fixed throughout the main text a set $X$ of variables:
terms, atoms, and languages may only contain variables in $X$, and similarly for closed expressions (beside recursion variables from $R$, which must be bound).

However, to perform proofs by induction on expressions and still be able to restrict ourselves to closed expressions, it is convenient to let this set $X$ vary. Indeed, when we encounter a closed expression of the shape $\mu x.e$, we usually want to use the induction hypothesis on $e$, which is not closed: it is an $x$-expression, and $x$ belongs to $R$. In those cases, we see $e$ as a closed expression with variables in $X\cup \set x$ and recursion variables in $R\backslash \set x$ (up to alpha conversion if $x$ was reused in fixpoint subexpressions of $e$---this is where we need $R$ to be countably infinite).
Formally, we let $X$ vary in such inductions: we prove properties of the shape ``for all sets $X$ and all closed expressions $e$ over $X$, \dots''.

\longnotations

For the sake of precision, we sometimes make this dependency on $X$ explicit in the following appendices, thus writing $\term X$, $\expr X$, $\atom X$ and $\lang X$ for the sets of terms, expressions, atoms, and languages with variables in $X$. We do not track the dependency on the set $R$ of recursion variables: it is irrelevant in most cases.

\shortnotations

\renewcommand\prooflink{}
\renewcommand\depE{}

\section{Proofs for \Cref{sec:framework} (Framework)}
\label{app:framework}
\printProofs[framework]

\section{Proofs for \Cref{sec:closure} (Closed languages)}
\label{app:closure}
\printProofs[closure]

\section{Proofs for \Cref{sec:axiomatisation} (Axiomatisations)}
\label{app:axiomatisation}
\printProofs[axiomatisation]

\section{Proofs for \Cref{sec:reductions} (Reductions)}
\label{app:reductions}
\printProofs[reductions]

\section{Proofs for \Cref{sec:examples} (Examples)}
\label{app:examples}
\printProofs[examples]

\section{General lattice theory}
\label{app:lattice}

In this last appendix, we restate and slightly generalise some of the general theory regarding complete lattices found in~\cite[Appendix~A]{pous2024tools}.

Throughout this section, we let $(\L, \bigvee)$ be a complete lattice.
For $x,y \in \L$, we write $x + y = \bigvee \{x,y\}$ and $x \leq y$ if and only if $x + y = y$.

\begin{theorem}[Knaster-Tarski, \cite{tarski1955lattice}]
  \label{thm:knaster-tarski}
  Every monotone function $f \colon \L \to \L$ admits a least fixpoint $\mu f$.
  This least fixpoint is also the least pre-fixpoint: for all $x\in L$, $fx\leq x$ implies $\mu f\leq x$.
\end{theorem}

Recall that the space of monotone functions $\L \to \L$ again forms a complete lattice. 
We write $1 \colon \L \to \L$ for the identity function on $\L$, and given monotone functions $f,g \colon \L \to \L$, we write $fg \colon \L \to \L$ for their composition.
A \emph{closure} is a monotone function $c \colon \L \to \L$ such that $1 \leq c$ and $cc \leq c$.
Given a monotone function $f \colon \L \to \L$, we write $f^* \colon \L \to \L$ for the least closure above $f$, which can be defined as the function $x\mapsto \mu y.x+fy$. 

For every monotone function $f\colon \L \to \L$, we have $f\bigvee S \geq \bigvee_{x\in S}f(x)$ for all subsets $S\subseteq \L$.
As mentioned in \Cref{sec:preliminaries}, a function $f\colon \L \to \L$ is \emph{continuous}\footnote{In \cite{pous2024tools}, those are called \emph{linear}, a term which we prefer to avoid in the present context.} if it preserves all joins: $f\bigvee Y=\bigvee_{x\in Y}f(x)$ for all subsets $Y\subseteq \L$. Every continuous function is monotone.

The following lemmas from~\cite[Appendix~A]{pous2024tools} make it possible to ease proofs of partial commutation between functions.
\begin{lemma}
  For all monotone functions $f,g,h$, we have
  \begin{align*}
    gf \leq fh \implies g^*f \leq fh^*\enspace.
  \end{align*}
\end{lemma}

\begin{lemma}
  \label{lem:fixed point respects leq}
  For all continuous functions $f$ and all monotone functions $g,h$, we have
  \begin{align*}
    fg \leq hf \implies fg^* \leq h^*f\enspace.
  \end{align*}
\end{lemma}

\begin{corollary}\label{lem:fixed point respects leq 2}
  For all continuous functions $f$, all monotone functions $g$, and all closures $c$, we have
  \begin{align*}
    fg \leq cf \implies fg^* \leq cf\enspace.
  \end{align*}
\end{corollary}

Next we generalise the notion of contextual function from \cite[Appendix~A]{pous2024tools}:
\begin{definition}\label{def:contextual}
  Let $f \colon \L^n \to \L$ be an $n$-ary function.
  A function $h \colon \L \to \L$ is \emph{contextual w.r.t. $f$} if it is monotone and
  for all $1 \leq i \leq n$ and for all $x_1,\dots,x_n\in\L$, we have:
  \begin{align*}
    f(x_1, \dots, h(x_i), \dots, x_n) \leq h(f(x_1,\dots, x_n))\enspace.
  \end{align*}
\end{definition}
When closures are contextual w.r.t. a monotone function, the above defining property can be strengthened as follows:
\begin{lemma}
    \label{lem:contextual characterization for closures}
    Let $f$ be a monotone $n$-ary function.
    A closure $c \colon \L \to \L$ is contextual w.r.t. $f$ if and only if for all $x_1,\dots,x_n\in\L$, we have
    \begin{align*}
      f(c(x_1), \dots, c(x_n)) \leq c(f(x_1,\dots, x_n))\enspace.
    \end{align*}
\end{lemma}
\begin{proof}
    If $c$ is contextual, then we have
    \begin{align*}
        f(c(x_1), c(x_2), c(x_3), \dots, c(x_n)) 
        &\leq
        c(f(x_1, c(x_2), c(x_3), \dots, c(x_n))) \\
        &\leq
        c^2(f(x_1, x_2, c(x_3), \dots, c(x_n))) \\
        &\leq
        \dots \\ 
        &\leq
        c^n(f(x_1, x_2, x_3, \dots, x_n)) \\
        &\leq
        c(f(x_1, x_2, x_3, \dots, x_n))\enspace.
    \end{align*}
    The converse implication comes from monotonicity of $f$, and $x\leq c(x)$.
\end{proof}
Finally, contextuality of $h^*$ w.r.t. continuous functions can be obtained from contextuality of $h$.
\begin{lemma}
    \label{lem:contextual implies fixed point contextual}
    Let $f \colon \L^n \to \L$ be an $n$-ary function which is continuous in each coordinate.
    If a function $h$ is contextual w.r.t. $f$, then so is $h^*$.
\end{lemma}
\begin{proof}
    We aim to apply \Cref{lem:fixed point respects leq} by fixing all coordinates of $f$ but one arbitrarily.
    Fix $x_1, \dots, x_{n-1} \in \L$ and define the function
    \begin{align*}
      \tilde{f} \colon x \mapsto f(x_1,\dots, x, \dots, x_{n-1})\enspace.
    \end{align*}
    Since $h$ is contextual w.r.t. $f$, we have that
    $
        \tilde{f}h \leq h \tilde{f}.
    $
    Since $f$ is continuous in each coordinate, $\tilde{f}$ is continuous and we can apply \Cref{lem:fixed point respects leq}.
    This shows $\tilde{f}h^* \leq h^* \tilde{f}$, hence
    \begin{align*}
      f(x_1,\dots, h^*(x), \dots, x_{n-1}) \leq h^*(f(x_1,\dots, x, \dots, x_{n-1}))\enspace.
    \end{align*}
    Since we fixed an arbitrary coordinate, we can conclude this holds for all coordinates, hence that $h^*$ is contextual.
\end{proof}

\end{document}